\documentclass[11pt]{article}
\usepackage[margin=1in]{geometry}
\usepackage{amssymb,amsfonts,amsmath,amsthm,amscd,dsfont,mathrsfs,bbold,pifont}
\usepackage{blkarray}
\usepackage{graphicx,float,psfrag,epsfig,color,yhmath}
\usepackage{comment}
\usepackage{subcaption}

\usepackage{algorithm}
\usepackage[noend]{algpseudocode}

\makeatletter
\def\BState{\State\hskip-\ALG@thistlm}
\makeatother

	\usepackage{microtype}
	\usepackage[pdftex,pagebackref=true,colorlinks]{hyperref}
	\hypersetup{linkcolor=[rgb]{.7,0,0}}
	\hypersetup{citecolor=[rgb]{0,.7,0}}
	\hypersetup{urlcolor=[rgb]{.7,0,.7}}

\newenvironment{psmallmatrix}
  {\left(\begin{smallmatrix}}
  {\end{smallmatrix}\right)}

\newcommand{\mF}{\mathbb{F}}

\newtheorem{theorem}{Theorem}

\newtheorem{lemma}{Lemma}
\newtheorem{corollary}{Corollary}
\newtheorem{definition}{Definition}
\newtheorem{remark}{Remark}

\newtheorem{conjecture}{Conjecture}

\begin{document}
\title{A proof that Reed-Muller codes achieve Shannon capacity\\ on symmetric channels}
\author{Emmanuel Abbe and Colin Sandon \\ EPFL}
\date{}
\maketitle

\begin{abstract}
Reed-Muller codes were introduced in 1954, with a simple explicit construction based on polynomial evaluations, and have long been conjectured to achieve Shannon capacity on symmetric channels. Major progress was made towards a proof over the last decades; using combinatorial weight enumerator bounds, a breakthrough on the erasure channel from sharp thresholds, hypercontractivity arguments, and polarization theory. Another major progress recently established that the bit error probability vanishes slowly below capacity. However, when channels allow for errors, the results of Bourgain-Kalai do not apply for converting a vanishing bit to a vanishing block error probability, neither do the known weight enumerator bounds. The conjecture that RM codes achieve Shannon capacity on symmetric channels, with high probability of recovering the codewords, has thus remained open. 

This paper closes the conjecture's proof. It uses a new recursive boosting framework, 
which aggregates the decoding of codeword restrictions on `subspace-sunflowers', handling their dependencies via an $L_p$ Boolean Fourier analysis, and using a list-decoding argument with a weight enumerator bound from Sberlo-Shpilka. The proof does not require a vanishing bit error probability for the base case, but only a non-trivial probability, obtained here for general symmetric codes. This gives in particular a shortened and tightened argument for the vanishing bit error probability result of Reeves-Pfister, and with prior works, it implies the strong wire-tap secrecy of RM codes on pure-state classical-quantum channels. 
\end{abstract}

\newpage

\tableofcontents
\newpage

\section{Introduction}
Shannon introduced in 1948 the notion of channel capacity \cite{Shannon48}, as the largest rate at which messages can be reliably transmitted over a noisy channel. In particular, for the canonical binary symmetric channel which flips every coordinate of a codeword independently with probability $\epsilon$, Shannon's capacity is $1-H(\epsilon)$, where $H$ is the binary entropy function. To show that the capacity is achievable, Shannon used a probabilistic argument, i.e., a code drawn uniformly at random. In the `worst-case' or `Hamming' setting \cite{Hamming50}, random codes can also be used to show that rates up to $1-H(2 \epsilon)$ are achievable\footnote{One should use $1-H(2 \epsilon) \wedge 1$ in case $\epsilon>1/2$.}, as codewords there must produce a packing of $\epsilon n$-radius balls (i.e,  a distance of $2\epsilon n$). Thus Shannon can reach higher rates due to rare error events being tolerated. 

Obtaining explicit code constructions achieving these limits has since then generated decades of research activity across electrical engineering, computer science and mathematics.   

\paragraph{A simple explicit code}
The Reed-Muller code was introduced by Muller in 1954 \cite{Muller54}, and is one of the first and simplest code. Reed developed shortly after a decoding
algorithm succeeding up to half its minimum distance \cite{Reed54}. The code construction can be described with a greedy  procedure. Consider the construction of a subspace of $\mF_2^n$ (a linear code) with blocklength given by a power of two, $n=2^m$. One can naturally start with the all-0 codeword. If one has to pick
a second codeword, then the all-1 codeword is the best choice under most relevant criteria. If one
has to now keep these two codewords, the next best choice to maximize the code distance is a codeword with half-0 and half-1, and to continue building a basis sequentially, one can add a few more vectors of weight $n/2$ to preserve a
relative distance of half, completing the Reed-Muller $RM(m,1)$ code of order $1$ (also called the Hadamard or augmented simplex code).  At this point, it may be less clear how to pick the next codeword,
but one can simply iterate the previous construction on the support of the previously picked
codewords, and re-iterate this a number of time after each saturation, reducing each time the distance by half. This gives the $RM(m,r)$
code, whose basis is equivalently defined by the evaluation vectors of monomials of degree at most $r$ on $m$ Boolean variables:
\begin{align}
RM(m,r)=\{(f(x_1),\ldots, f(x_{2^m})): x_1,\ldots, x_{2^m} \in \mathbb{F}_2^m, f \in \mathbb{F}_2(X_1,\ldots,X_m), \mathrm{deg}(f)\le r\}.
\end{align}

As mentioned, the first order RM code $RM(m,1)$ is the augmented simplex code or Hadamard code. The simplex code is the dual of the Hamming code that is `perfect', i.e., it provides a perfect sphere-packing of the space (it achieves the sphere-packing bound). This strong property is clearly lost for general RM codes, but RM codes preserve nonetheless a decent distance (at root block length for constant rate). In the worst-case model, where the distance controls the error-decoding capabilities, the RM code does not give a `good' family of codes (i.e., a family of codes with asymptotically constant rate and constant relative distance). Other codes can achieve better distance/rate trandeoffs such as expander codes \cite{expander}. However, once put under the light of Shannon's probabilistic error model, for which the minimum distance is no longer the right figure of merit, RM codes perform surprisingly well,  even as well as random codes potentially.  

\subsection{The conjecture}
It has long been conjectured that RM codes achieve Shannon capacity on the binary symmetric channel (BSC), and more generally  binary-input memoryless symmetric (BMS) channels (also called here simply `symmetric channels'). We refer to Section \ref{bms} for the formal definition of BMS channels; for now, it is sufficient to consider the BSC, the main case of interest, as BMS are mixtures of BSCs.

\begin{conjecture}
{\it For any BMS channel $\mathcal{P}$, and any rate $R$ below the capacity of $\mathcal{P}$ given by $C(\mathcal{P})=I(U,\mathcal{P})=(1/2)\sum_{x \in \{0,1\}, y \in \mathcal{Y}}\mathcal{P}(y|x)\log_2 \mathcal{P}(y|x)/(\mathcal{P}(y|0)/2+\mathcal{P}(y|1)/2)$,  where\footnote{This the mutual information of the channel with the uniform input distribution, which is the optimal input distribution due to the symmetry of BMS channels.} $\mathcal{Y}$ is the output alphabet of $\mathcal{P}$ (with $C(\mathcal{P})=1-H(\epsilon)$ when $\mathcal{P}=\mathrm{BSC}(\epsilon)$), a sequence of RM$(m_i,r_i)$ codes of rate $R_i={m_i \choose \le r_i}2^{-m_i}$ tending to $R$ can be decoded successfully with high probability, i.e., any of the $\lfloor 2^{nR_i} \rfloor$ codewords can be decoded with probability $1-o_{m_i}(1)$ despite corruptions from $\mathcal{P}$.} 
\end{conjecture}

It is hard to track back the first appearance of this belief in the literature, but \cite{Kudekar16STOC} reports that it was likely already present in the late 60s.
The claim was mentioned explicitly in a 1993 talk by
Shu Lin, entitled ?RM Codes are Not So Bad? \cite{Lin93}. A 1993 paper by Dumer and
Farrell also contains a discussion on the matter \cite{Dumer93}, as well as the 1997 paper of Costello and Forney on the `road to channel capacity' \cite{Costello07}. Since
then, the topic gathered significant attention, and the activity sparked with the emergence of polar codes \cite{Arikan09}; Ar\i kan mentioned this as one of the major open problems in coding theory at ITW Dublin in 2010. Due to the broad  relevance\footnote{RM codes with binary or q-ary fields have been used for instance in cryptography \cite{Shamir79,BF90,Gasarch04,Yekhanin12}, pseudo-random generators and randomness extractors \cite{DBLP:journals/jcss/Ta-ShmaZS06,bogdanov-viola}, hardness amplification, program testing and interactive/probabilistic proof systems \cite{BFL90,Sha92,ALMSS98}, circuit lower bounds \cite{Razborov}, hardness of approximation \cite{barak2012making,BKSSZ}, low-degree testing \cite{DBLP:journals/tit/AlonKKLR05,DBLP:journals/siamcomp/KaufmanR06,DBLP:journals/rsa/JutlaPRZ09,DBLP:conf/focs/BhattacharyyaKSSZ10,DBLP:journals/siamcomp/HaramatySS13}, private information retrieval \cite{DBLP:journals/jacm/ChorKGS98,DBLP:journals/jacm/DvirG16,DBLP:journals/jacm/ChorKGS98,BEIMEL2005213,DBLP:conf/focs/BeimelIKR02}, and compressed sensing  \cite{Calderbank2010reed,Calderbank10,Barg15}.} of RM codes in computer science, electrical engineering and mathematics, the activity scattered in a wide line of works \cite{dumer1,dumer2,dumer3,Carlet05,hell,arikan-RM,Arikan09,Arikan2010survey,Kaufman12, Abbe15stoc,Abbe15,Kudekar16STOC,Kudekar17,Mondelli14,Saptharishi17,AY18,YA18,Sberlo18,comparingBitMAP,Samorod18old,Samorod18,abbe2020almost,HSS,Lian20,Fathollahi21,ASY21,Reeves21,Geiselhart21,BhandariHS022,RM_fnt}; see also \cite{RM_fnt}.

\paragraph{Relations to polar codes} 
In a breakthrough paper, Ar\i kan showed that the explicit class of polar codes achieve Shannon capacity on any BMS channel \cite{Arikan09,ArikanTelatar}. 
Given the close relationship between polar and RM codes, the belief that RM codes could also be proved to achieve capacity on BMS channels intensified. 
Polar codes are derived from the same square matrix, i.e., the matrix whose rows correspond to evaluations of monomials, which can also be expressed as $
G_n:=\begin{psmallmatrix}1 & 1\\0 & 1\end{psmallmatrix}^{\otimes m}
$, but polar codes use a different row selection that is channel dependant. One can view the difference as follows: apply $G_n$ (over $\mathbb{F}_2$) to a vector with i.i.d.\ Bernoulli$(\epsilon)$ components in order to compress the noise vector of the BSC channel; polar codes keep the rows having largest `conditional entropy' in the transformed vector, proceeding top to down. Using `conditional' entropies is a notable constraint of polar codes; it gives significant benefit in terms of decoding complexity, it is also convenient to bound the error probability (it is an `average-case' measure) but it may not be (and is not) the best way to proceed for performance. Instead, RM codes apply a more `macroscopic' and `universal' rule: keep the rows having largest Hamming weights. This is a natural choice for the worst-case model (in fact one can show that RM codes are the best subcode of $G_n$ in terms of distance), but this may not necessarily be natural in the Shannon setting. It remains reasonable to expect that the densest rows will extract the noise randomness effectively\footnote{This insight initiated the collaboration with A. Wigderson and A. Shpilka \cite{Abbe15stoc}.}, but this is much harder to establish as it uses a `worst-case measure' for an `average-case' problem.

\subsection{Progress and related references} We now overview the main progress on the conjecture in rough chronological order. 
\begin{enumerate}

\vspace{-.1cm}
\item \cite{sidel} is one of the first work establishing performance guarantees for RM codes in the Shannon setting, but for fixed $r$ and with an offset compared to the capacity scaling. 
A first capacity-achieving scaling is obtained in \cite{hell} for the special case of $r=1,2$. During that time, a line of work by Dummer also makes major progress on the RM code recursive decoding \cite{dumer1,dumer2,dumer3}, albeit without capacity-achieving results. The  polar codes breakthrough then takes place in \cite{Arikan09}, reviving the   conjecture as discussed earlier.  

\vspace{-.1cm}
\item The first capacity-achieving result for non-constant $r$ appears in \cite{Abbe15stoc,Abbe15}, which develops the weight enumerator bound approach for both the BEC and BSC channels, leveraging techniques from the key paper \cite{Kaufman12}. Prior to \cite{Kaufman12}, a line of work in the 70s derived weight enumerator bounds for RM codes, from the original work of Sloane and Berlekamp for $r=2$ \cite{sloane-RM} to twice (and 2.5) the minimum distance \cite{kasami1,kasami2}. However, no notable progress had appeared since 1976, until \cite{Kaufman12} managed to break through the 2.5 barrier and provide bounds for distances in a much broader regime of small $r$, and in particular, introducing the epsilon-net approach. One contribution of \cite{Abbe15stoc} is to expand and refine this approach to larger regimes of $r$, possibly linear in $m$. 

\vspace{-.1cm}
\item A breakthrough then took place for the binary erasure channel (BEC), with the celebrated papers \cite{Kudekar16STOC,Kudekar17}. The BEC is the simplest non-trivial channel; it does not allow for errors but only erasures. The BEC benefits in particular from a strong property: the event of decoding a bit incorrectly given the other noisy bit is a monotone property on the set of erasure patterns (the exit function). This allows \cite{Kudekar16STOC} to connect to generic results from sharp thresholds of monotone properties \cite{Friedgut96,Bourgain97} and conclude with an elegant proof that any code having from enough symmetries (namely, doubly transitivity), is capacity achieving on the BEC.  Further, due to the results of \cite{Bourgain97}, under such strong symmetries, the threshold must be in fact sharp enough, which implies that the `local' analysis showing that the exit function is vanishing is sufficient to  obtain a vanishing block error probability. 

\vspace{-.1cm}
\item In \cite{Sberlo18}, an important improvement of \cite{Abbe15stoc} is obtained, with a similar proof technique but with tighter and more generic estimates, producing useful weight enumerator bounds for larger $r$, in particular the regime of $r$ inducing constant rates, of interest for the conjecture. As we shall see, this improvement will turn out to be key in one step of our proof. 

\vspace{-.1cm}
\item While polar codes have revived the activity on the RM code conjecture, due to their similarity, the prior progress is based on techniques that have no apparent relations to polar codes. 
In \cite{AY18}, a connection to polarization theory is developed, showing that by decoding bits sequentially in the RM code order (by decreasing weight), the conditional entropy also polarizes. This in turn implies that the code obtained by retaining the `large-entropy' components, defined as the `twin-RM code', must be capacity achieving on any BMS channel. To close the conjecture, it remains to show that this `twin-RM' code is equivalent to the RM code. There appears to be sufficient evidence supporting this --the two codes are equivalent for the BEC, and numerical simulations as well as partial theoretical results \cite{abbe2020almost} support it--but that program remains to be closed, with intriguing connections to additive combinatorics \cite{tao_vu_2006,lovett_book,sanders}. 

 \vspace{-.1cm}
\item A new analytical approach was developed in \cite{Samorod18old,Samorod18} to obtain bounds for the weight enumerator of RM codes in new parameter regimes. In particular, \cite{Samorod18} allows to leverage results for erasure channels to new bounds on the weight enumerator, via hypercontracivity arguments. Improvements on how to exploit the erasure duality were further developed in \cite{HSS} which achieves an important milestone: showing that RM codes can be decoded at constant rate over the BSC channel. While this does not give a capacity result, it settled an open problem from \cite{Abbe15stoc} as techniques therein did not seem able to capture this regime. 

\vspace{-.1cm}
\item In \cite{YA18}, a recursive projection aggregation algorithm is developed, decoding a codeword by recursively projecting it on subspaces, and aggregating the projection decodings with a majority vote. While no capacity results are obtained for this algorithm, which is also limited to low $r$, some of the high-level ideas of recursively combining sub-components decoding of the codewords have a common flavor with the generic boosting procedure developed here. However, we rely here not on projections but restrictions (which do not suffer from noise increase) and on the carefully chosen `sunflower subspaces', to aggregate effectively the restrictions and to allow for a controlled dependency analysis. 

\vspace{-.1cm}
\item Finally, another major result recently showed that the exit function or bit error probability vanishes below capacity for RM codes on any BMS channel. The approach of \cite{Reeves21} is based on the area theorem and exit function as in \cite{Kudekar16STOC}, working out more tailored bounds for BMS channels based on an mmse version of the exit function and the nesting property of RM codes.  While \cite{Reeves21} comes close to establishing the conjecture, due to the very slow decay obtained for the exit function, the dominant part of the conjecture appears to be left open as discussed next.  
\end{enumerate}


We refer to \cite{ASY21,RM_fnt} for further references. 
\vspace{-.1cm}
\subsection{From bit to block and the critical part of the proof}\label{critical}
As mentioned above, a recent paper \cite{Reeves21} showed that at any constant rate below capacity, RM codes achieve a vanishing `bit-error probability' on all BMS channels, with a bit error scaling at the rate of $O(\log\log(n)/\sqrt{\log(n)})$. 
While this is an important result, such a slow rate does not imply that RM codes achieves capacity in the classical sense of recovering the codewords, neither from a union bound nor the current weight enumerator bounds as discussed next. In fact, with our boosting framework, such a slow decay allows to jump the first step but is otherwise not much further in the progress compared to a constant but non-trivial decay, which we can establish here more directly and more generically as our base case. Most of the work then relies on reaching a fast-enough decay to be able to recover the codewords with list-decoding, and this requires a much faster decay. 

Let us first clarify the different types of decoding requirement that one can study for the channel coding problem (assuming a BSC). The Shannon capacity is classically defined \cite{csiszar_book,cover_book,gallager_book} as 
the supremum over all rate $R$ for which there exists an encoding map $E^n: [ \lfloor 2^{nR} \rfloor ] \to \mathbb{F}_2^n$ and a decoding map $D^n: \mathbb{F}_2^n \to [ \lfloor 2^{nR} \rfloor ]$ such that for any message $m \in \lfloor 2^{nR} \rfloor$, $\mathbb{P}(m \ne D^n(E^n(m) \oplus Z^n))=o_n(1)$. Thinking in terms of codewords, i.e., denoting by $X^n$ the map $E^n$ and $\hat{X}^n$ the map that returns the codeword corresponding to the decoded message, i.e., $\hat{X}^n(\cdot)=E^n(D^n(\cdot))$, the previous message error probability is equal to $\mathbb{P}(X^n(m) \ne \hat{X}^n(X^n(m)\oplus Z^n))$.  
Thus the probability of error applies to the messages, or equivalently, to the entire codewords, and is usually called the `block' error probability or simply error probability. Since one can also consider $m$ to be drawn uniformly at random, adding this randomness to the probability\footnote{RM codes have enough symmetries that any codeword behaves similarly for error events.
}, one can simplify notation of the block error to $\mathbb{P}(m \ne \hat{m})$ or $\mathbb{P}(X^n \ne \hat{X}^n)$ (where $X^n$ is drawn uniformly at random among the codewords).

One can also define the `bit' error probability, as the probability that a coordinate in a codeword is incorrectly decoded, i.e., $\mathbb{P}(X_i \ne \hat{X}_i(X^n\oplus Z^n))$ (which is $i$ independent for RM codes). As will be discussed next, one may also consider the exit function variant where one removes the noisy observation of the bit to be decoded, using the decoder based on the other bits, $\mathbb{P}(X_i \ne \hat{X}_1(X_{-i}^n\oplus Z_{-i}^n))$ where $X_{-i}^n$ is the $(n-1)$-dimensional vector obtained by removing component $i$ from $X^n$.

The bit-error probabilty is sometimes used and desirable to establish a stronger\footnote{A `strong' converse also refers to an error probability tending to 1 above capacity.} converse, i.e., showing that one cannot even reach a vanishing bit-error probability above capacity (as shown here and in \cite{Reeves21} as well), while the block-error probability is desirable and commonly used to establish the achievability part of the capacity theorem, to recover the actual messages. 

If the bit-error probability vanishes at a rate $o(1/n)$, then a simple union-bound allows to imply that the block-error probability is also vanishing. However, it is a challenge to establish this scaling for RM codes and \cite{Reeves21} obtains a scaling far slower, i.e., $O(\log\log(n)/\sqrt{\log(n)})$. A weaker scaling than $o(1/n)$ could still suffice for a slightly more elaborate upper-bound that uses weight enumerator bounds such as in \cite{Sberlo18}, with a list-decoding argument, as discussed for instance in \cite{comparingBitMAP}. However, the required scaling for such a more elaborate technique is again far from $O(\log\log(n)/\sqrt{\log(n)})$, in paricular \cite{comparingBitMAP} requires $n^{-\Omega(1)}$. 

In the case of the BEC, one can take a much more direct path to go from the bit to the block error probability, because the BEC channel can be framed  in the study of a monotone Boolean formula, and under enough  symmetry (granted here by the RM code), general results from Bourgain and Kalai \cite{Bourgain97} have derived bounds on the critical window scaling, which imply directly that the bit-error probability must be vanishing as $o(1/n)$ (thus giving a vanishing block error from a union bound). We refer to \cite{Kudekar16STOC} for the details.  Unfortunately, for the BSC channel or other BMS channels, one cannot cast the error events as Boolean monotone properties when studying the bit (or block) error probabilities, and thus one cannot import directly results from sharp thresholds for monotone properties. In fact, besides being formally not `Boolean', the error events are non `monotone' (for potential extensions to non-Boolean properties). This is why additional machinery is already required to show that the exit function on the BSC channel is at least vanishing in \cite{Reeves21}. 

Establishing a fast bit-error probability decay is thus highly non-trivial. Our proof allows to get there with a collection of steps that hold each other tightly. 
In fact, we do not prove directly a $o(1/n)$ decay, but can obtain it with 5 intermediate levels of our boosting framework (cf.\ Figure \ref{boosting_blocks_1}): 
\begin{enumerate}

\item We first rely on a weak decoding criteria for channel coding; this is not a vanishing exit probability, but a constant and non-trivial probability. I.e., using generic entropic arguments, we first show that for any code operating below capacity that is symmetric, the probability of estimating a bit given the other noisy bits is $1/2-\Omega(1)$;
\item We then boost this to $O(2^{-\sqrt[3]{m}})$ using a one step `sunflower-boosting' and $L_1,L_2$ Boolean Fourier estimates. This improves the $O(\log(m)/\sqrt{m})$ scaling of \cite{Reeves21} (and gives a shorter proof) but it is still far from the $o(1/ 2^m)$ needed for a union bound, or from the scaling needed for a list-decoding step using \cite{Sberlo18};

\item We further boost this to $2^{-\Omega(\sqrt{m}\log(m))}$ using `recursive sunflower-boosting'; this is where most of the technical work lies, in particular the $L_4$ Boolean Fourier estimates;

\item We then combine our last estimate with a list decoding argument, using a weight enumerator bound from \cite{Sberlo18} to obtain a vanishing block error probability; this requires a similar argument as in \cite{comparingBitMAP}, although the theorem therein stated for a stronger scaling. Also, we do not only need the improvement of \cite{Sberlo18} over \cite{Abbe15stoc}, but also an intermediate step in the proof of Theorem 1.2 from \cite{Sberlo18}; 

\item We further boost this block error probability to $2^{-2^{\Omega(\sqrt{m}})}$ using `grid-boosting', which implies, a fortiori, a bit-error probability of $2^{-2^{\Omega(\sqrt{m}})}$ (which is in particular $o(1/n)$). 
\end{enumerate}

\section{Main result}
\begin{theorem}\label{main_thm}
Let $\mathcal{P}$ be a BMS channel and let $C(\mathcal{P})$ be its capacity, i.e., $C(\mathcal{P})=I(U,\mathcal{P})=(1/2)\sum_{x \in \{0,1\}, y \in \mathcal{Y}}\mathcal{P}(y|x)\log_2\left( \mathcal{P}(y|x)/(\mathcal{P}(y|0)/2+\mathcal{P}(y|1)/2)\right)$ where $\mathcal{Y}$ is the output alphabet of $\mathcal{P}$ ($C(\mathcal{P})=1-H(\epsilon)$ when $\mathcal{P}=\mathrm{BSC}(\epsilon)$). Let $\{m_i\}_{i \ge 1}$ and $\{r_i\}_{i \ge 1}$ be sequences of positive integers such that $r_i \le m_i$ for all $i$ and $\lim_{i\to\infty} m_i=\infty$. 
\begin{itemize}
\item If $\lim\sup_{i\to\infty} {m_i \choose \le r_i}2^{-m_i}=R < C(\mathcal{P})$, then the block error probability of $RM(m_i,r_i)$ vanishes as $2^{-2^{\Omega(\sqrt{m_i}})}$, i.e., one can recover with high probability any RM codeword from its corruption. 

\item If $\lim\inf_{i\to\infty} {m_i \choose \le r_i}2^{-m_i}=R > C(\mathcal{P})$, then the bit error probability of $RM(m_i,r_i)$ blows up as $1/2-2^{-\Omega(\sqrt{m})}$, i.e., no single component of a RM codeword can be decoded\footnote{This holds for the optimal maximum likelihood or MAP decoder.} from its corruption with asymptotically non-trivial probability. 
\end{itemize}
\end{theorem}
Some implications of this theorem: 
(i) This shows that RM codes, despite having a simple explicit construction of encoding complexity $O(n \log_2 n)$ that is universal for all BMS channels of same capacity, can match the performance of Shannon's random codes in terms of achieving capacity with maximum likelihood decoding. (ii) As discussed in \cite{Reeves21}, our results have also consequences in the quantum setting. It shows that Quantum Reed-Muller codes achieve the hashing bound from \cite{wilde_2013}, and since our result holds for a vanishing block error probability, it implies with the duality result of \cite{renes} that RM codes achieve strong secrecy on classical-quantum channels for both the BSC and pure-state wire-tap channels.

Theorem \ref{main_thm} part (i) and (ii) are proved in Sections \ref{recovering} and \ref{converse} respectively, for the BSC, and Section \ref{bms} provides the adjustment needed to extend the proofs to BMS channels.

\begin{figure}
    \centering
    \includegraphics[width=1\textwidth]{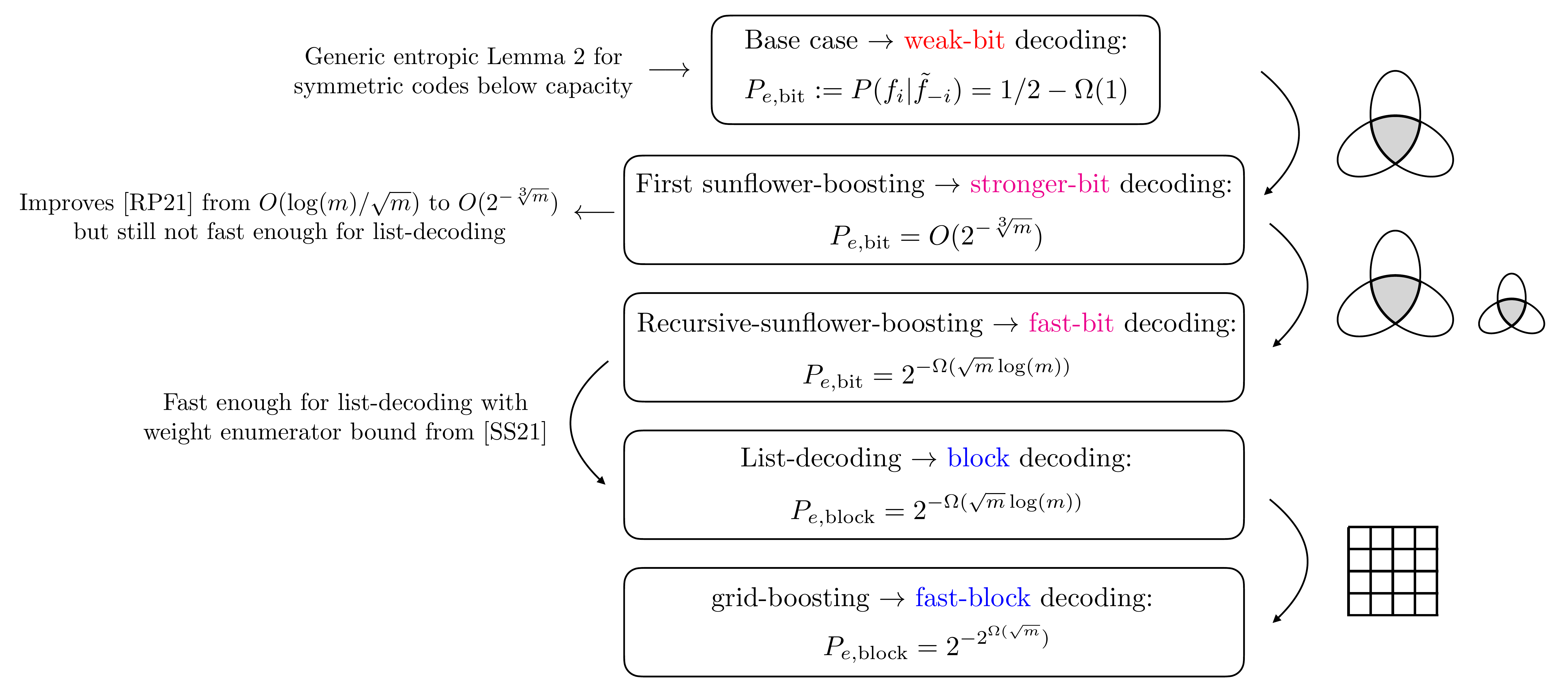}
    \caption{The various improvements of the local to global error bounds in our boosting framework.}
    \label{boosting_blocks_1}
    \vspace{-.5cm}
\end{figure}

\section{Proof technique}\label{technique}
We will prove increasingly strong bounds on local and global error events of RM codes. 

{\bf Recursive boosting decoding algorithm.} Most of the work relies on providing a method to convert any algorithm that decodes a codeword coordinate in $RM(m,r)$ with non-trivial but otherwise weak success probability, to an algorithm having a higher and controllable success probability for a codeword component in $RM(\overline{m},r)$ with an $\overline{m}$ slightly larger than $m$.

{\bf Base case.} Our base case is a weak bound on the probability of error when decoding a bit when observing the other noisy bits (i.e., the exit function). However, this bound only needs to be of order $1/2 - \Omega(1)$, which is much weaker than a vanishing bound for the exit function as proved in \cite{Reeves21}. In fact, we can 
obtain such a weak bound generically, for any code below capacity having sufficient symmetries, with an argument that relies on entropic inequalities (Lemma \ref{baseErrorGen}).

{\bf Constructing a subspace sunflower.\footnote{We use the `sunflower' terminology from \cite{sunflower_1,sunflower_2}, but do not need the ``sunflower lemma'' \cite{sunflower_1}.}} 
From this base case, we start boosting the weak error bound to better bounds. To develop the boosting technique, we consider a subspace $V$ of $\mathbb{F}_2^{\overline{m}}$ containing the bit to be decoded, where $V$ is chosen arbitrarily but independently of $Z$ and with dimension $\underline{m}$, such as  take $V=\mathbb{F}_2^{\underline{m}} \times 0^{\overline{m}-\underline{m}}$. 
Then 
we use a large number of $m$-dimensional subspaces of $\mathbb{F}_2^{\overline{m}}$, $W_1$,...,$W_b$, that all intersect at $V$ and only at $V$. We use a greedy algorithm to construct these `petal' subspaces $W_1$,...,$W_b$ (cf.\ Lemma \ref{sunflowerLem}). We next independently attempt to decode the bit in question from the restriction of the noisy codeword to each of these petals, ignoring the noisy value of the target bit in order to mitigate the dependencies between these decodings. 
At the end we take a majority vote on the petals to decode the original bit. The advantage of working on the petals is that we have multiple of them and their overlap is controlled by the intersection\footnote{We use an overlapping set $V$ greater than $0^m$ as any two subspaces of dimension greater than $m/2$ intersect somewhere other than $0^m$, and spaces of dimension $m/2$ are too small to decode on effectively.} $V$, which allows us to manage the dependencies and gain concentration in the majority vote analysis.

{\bf Analyzing the conditional probabilities on the petals.} The main object of study then becomes the probability that we decode wrongly a bit conditioned on the restriction of the noise to a petal $W_i$ taking a fixed value. In order to discuss this a little more formally, consider the case of the binary symmetric channel that flips bits independently with probability $\epsilon$, and define $P_e(m,r,\epsilon)$ to be the error probability of decoding a bit of a codeword in $RM(m,r)$ when observing the other noisy bits. We denote by $Z\in\mathbb{F}_2^{2^m}$ the noise vector indicating which bits are flipped. For any $i\ne j$, the intersection of $W_i$ and $W_j$ is $V$, so the restrictions of $Z$ to the $W_i$ are independent\footnote{The restrictions of the actual code to the $W_i$ are not independent conditioned on the restriction of the code to $V$, but that does not matter because the probability of decoding the target bit incorrectly depends only on $Z$.} 
conditioned on the restriction of $Z$ to $V$. That in turn means that whether or not they decode the target bit correctly are independent conditioned on the restriction of $Z$ to $V$. Therefore if the restriction of $Z$ to $V$ is ``good", in the sense that the probability of decoding the restriction of the code to each of the $W_i$ in a way that gets the target bit wrong conditioned on that restriction is at most $1/2-\Omega(1)$, and if $b$ is large enough, we will almost certainly determine the value of the target bit correctly. 

We thus need to show that the restriction of $Z$ to $V$ is sufficiently likely to be good. Let $Q_{\underline{m},m}(z')$ be\footnote{We drop the parameters $r, \epsilon$ as they become implicit from now on.} the probability that we decode wrongly the bit $0^m$ by observing the other noisy bits, conditioned on the restriction of $Z$ to an $\underline{m}$-dimensional subspace ($V=\mathbb{F}_2^{\underline{m}} \times 0^{m-\underline{m}}$) being $z'$. We now want to bound $\mathbb{P}_{z'}[Q_{\underline{m},m}(z')>1/2-c]$ for some suitable constant $c$. A little more specifically, early in the boosting process we will have an upper bound on $\mathbb{E}_{z'}[Q_{\underline{m},m}(z')]=P_e(m,r,\epsilon)$ that is slightly less than $1/2$ and we will need to prove that $\mathbb{P}_{z'}[Q_{\underline{m},m}(z')>(\mathbb{E}_{z'}[Q_{\underline{m},m}(z')]+1/2)/2]=o(1)$ in order to prove that $P_e(\overline{m},r,\epsilon)=o(1)$ (first case). Later on in the process, we will have fairly tight bounds on $P_e(m,r,\epsilon)$ and we will need to prove that $\mathbb{P}_{z'}[Q_{\underline{m},m}(z')>1/3]$ is very small in order to prove sufficiently tighter bounds on $P_e(\overline{m},r,\epsilon)$ (second case).

{\bf Boolean Fourier analysis of $Q$.} Since $\mathbb{P}_{z'}[Q_{\underline{m},m}(z')>1/3]\le 9 \mathbb{E}_{z'}[Q^2_{\underline{m},m}(z')]$ and $\mathbb{P}_{z'}[Q_{\underline{m},m}(z')>(\mathbb{E}_{z'}[Q_{\underline{m},m}(z')]+1/2)/2]\le\frac{16}{(1-2\mathbb{E}_{z'}[Q_{\underline{m},m}(z')])^2}Var_{z'}[Q_{\underline{m},m}(z')]$, we rely on the Fourier transform of $Q$ to estimate its $L_1$ and $L_2$ norms. For any given $z'$, the value of $Q_{\underline{m},m}(z')$ is the expectation of $Q_{m,m}(z'')$ over all $z''\in\mathbb{F}_2^m$ having a restriction to $\mathbb{F}_2^{\underline{m}}$ that is $z'$. This implies that every term in the Fourier transform of $Q_{m,m}$ having support in $\mathbb{F}_2^{\underline{m}}$ is also in the Fourier transform of $Q_{\underline{m},m}$, while those having supports not in $\mathbb{F}_2^{\underline{m}}$ cancel themselves out when the averaging is applied. Also, $\mathbb{E}_{z'}[Q^2_{m,m}(z')]\le \mathbb{E}_{z'}[Q_{m,m}(z')]=P_e(m,r,\epsilon)$. So, in order to show that $P_e(\overline{m},r,\epsilon)$ is small compared to $P_e(m,r,\epsilon)$, it will suffice to show that
most of the weight on nonconstant terms is on terms that are not supported in $\mathbb{F}_2^{\underline{m}}$ using weaker assumptions (first case), 
 and that most of the weight is on terms that are not supported in $\mathbb{F}_2^{\underline{m}}$ using stronger assumptions (second case). 

{\bf Using affine invariance to reduce to non-concentration on subspaces.} Next, we observe that from the affine-invariance of RM codes, if $S\subseteq\mathbb{F}_2^m$, and $S'$ is the set resulting from applying a linear transformation to $S$, then the coefficients of the terms supported by $S$ and $S'$ in the Fourier transform of $Q_{m,m}$ must be the same. We also observe that if the subspace of $\mathbb{F}_2^m$ spanned by $S$ is $d$-dimensional then the fraction of the elements of the orbit of $S$ under linear transformations that are contained in $\mathbb{F}_2^{\underline{m}}$ is at most $2^{-d(m-\underline{m})}$. Every nonconstant term in the Fourier transform of $Q_{m,m}$ has a support that spans a subspace of dimension at least $1$,\footnote{Here is where the requirement that our attempt to decode a bit not use the noisy version of that bit comes in. If we allowed using the noisy version of the target bit the transform would have a term with support $\{0^m\}$. } which is enough to show that $Q_{\underline{m},m}$ is close enough to constant that most $z$'s are good for $m$ slightly larger than $\underline{m}$. This will imply a vanishing bit-error probability, scaling as $O(2^{-\sqrt[3]{m}})$, improving on the $O(\log(m)/\sqrt{m})$ from \cite{Reeves21} but not sufficiently to proceed with list-decoding. 
In order to establish tight enough bounds on $P_e(\overline{m},r,\epsilon)$ for the next phase, we will need to show that most of the weight is on terms having supports that are not contained in low dimensional subspaces, which is more involved. 

{\bf Non-concentration on subspaces via $L_4$-norm.} 
We next prove that for any low-dimensional set $S$, the Fourier energy of $Q_{m,m}$ on terms supported by linear transformations of $S$ is relatively small. For any given value of $z'\in\mathbb{F}_2^m$ and conditioned on $Z=z'$, one will always get the value of the target bit either right, wrong or both possible values of the target bit will be equally likely (tie case). Thus $Q_{m,m}(z')\in\{0,1/2,1\}$ for all $z'$. We know that $\mathbb{E}_{z'}[Q_{m,m}(z')]=o(1)$, so that implies that $Q_{m,m}$ is concentrated on a small fraction of the inputs. For any set $S$ spanning a low dimensional space, we will prove that the sum of the Fourier basis elements supported by linear transformations of $S$ is not concentrated on a small fraction of the inputs by bounding its $L_4$ norm and comparing that to its $L_2$ norm. That implies that the orbit of any given low-dimensional set cannot account for too much of the Fourier weight of $Q_{m,m}$, and there are at most $2^{2^d}$ different orbits of $d$-dimensional sets under linear transformations, so the combined weights of those terms will be small. 

{\bf Repeat.} Previous results allow us to show that $P_e(\overline{m},r,\epsilon)$ is small relative to $P_e(m,r,\epsilon)$ for $\overline{m}$ relatively close to $m$. Applying these bounds repeatedly, we show that $P_e(m,r,\epsilon)=2^{-\Omega(\sqrt{m}\log(m))}$ for the values of $m$, $r$, and $\epsilon$ that we care about.

{\bf Technical justifications.} There are a few details of this argument that are important to establishing that the boosted error bounds decrease fast enough. In the first phase of boosting, our bound on the error rate decreases exponentially in $\overline{m}-m$. That is not fast enough to get the bounds we need for our argument establishing that we can recover the codeword, so we need a second phase of boosting using subspace sunflowers. Similarly, it would be simpler to just always have all of the $W_i$ with dimensionality one greater than that of $V$, but then we would still need to have $\overline{m}-m=\Omega(\log(m))$ in order to limit the risk that the majority of the decodings of the $W_i$ get the target bit wrong despite the restriction of $Z$ to $V$ being good, and our bound on $P_e(\overline{m},r,\epsilon)$ would be roughly $P_e(m,r,\epsilon)/\sqrt{m}$, which would tighten too slowly. 

{\bf Bit-to-block decoding with tight list-decoding.} The next part is more standard, but requires a tight control of the bounds to bind the pieces together. We can apply our previous bound to every bit in order to obtain a decoding that gets at most $2^{m-\Omega(\sqrt{m}\log(m))}$ bits wrong with high probability. To determine exactly what the codeword was, we use list decoding, i.e., we make a list of every codeword that is within distance $2^{m-\Omega(\sqrt{m}\log(m))}$ of the previous decoding and then decode to the list element that is closest to the corrupted codeword. We then use a union bound over all codewords that are within distance $2\cdot 2^{m-\Omega(\sqrt{m}\log(m))}$ of the true codeword. By an intermediate\footnote{Theorem 1.2's bound in \cite{Sberlo18} does not suffices to prove our bound, but the intermediary bound at page 17 does.} result in the proof of Theorem 1.2 of \cite{Sberlo18}, the number of elements of $RM(m,r)$ that disagree with any codeword on at most $2^{m-\l}$ elements is at most 
\[2^{\sum_{j=\l}^r 17m(j-1)(j+2)+17(j+2){m-(j-1) \choose \le r-(j-1)}}\]
which is $2^{O(\l 2^{m-\l} \left(4(r-\l)(m-r)/(m-\l)^2\right)^{m/2}) }$ for the values of $\l$ that we care about. For a given codeword that disagrees with the true codeword on at least $2^{m-\l-1}$ bits, the probability that the corrupted code has greater agreement with it than with the true codeword is at most $(4\epsilon(1-\epsilon))^{2^{m-\l-2}}$ (Bhattacharyya bound). 
These last steps are similar to the arguments made in \cite{comparingBitMAP}, although their theorem is stated to require a much tighter bound on the bitwise error (namely $n^{-\Omega(1)}$) than allowed here. 
The probability that there is a codeword within $2\cdot 2^{m-O(\sqrt{m}\log(m))}$ of the true codeword that is more likely is then $\sum_{\l= \Omega(\sqrt{m}\log(m))}^m 2^{O(\l 2^{m-\l} \left(4(r-\l)(m-r)/(m-\l)^2\right)^{m/2}) } (4\epsilon(1-\epsilon))^{2^{m-\l-2}} \ll 2^{-\Omega(\sqrt{m}\log(m))}$.
This shows that we can recover the codeword with probability $1-2^{-\Omega(\sqrt{m}\log(m))}$. 

{\bf Boosting the block-error probability via high-dimensional grid boosting.} In order to establish tighter bounds on the block-error probability, we first observe that for sufficiently small $c$, any set $S\subseteq\{0,...,m\}$ with $|S|=c\sqrt{m}$ and any $x\in\{0,1\}^S$, the restriction of the code to the bits whose variables with indices in $S$ take on the values specified by $x$ is still an RM code with rate below the threshold. So, we can determine the restriction of the codeword to these bits from its noisy version with probability $1-2^{-\Omega(\sqrt{m}\log(m))}$. Furthermore, whether or not we can recover the restriction of the codeword to the bits whose coordinates in $S$ take on the values specified by $x$ from its noisy version, is independent of the corresponding event for another $x'$. With very high probability, it will thus be the case that most values of $x$ are ``good", where ``good'' means here that we can recover the corresponding restrictions of the codeword from their noisy versions. 

The advantages of doing this additional grid-boosting is that we can virtually guarantee that it will not get too many bits wrong, and the bits that it does get wrong will typically be arranged in patterns that are more convenient to handle. Specifically, we partition $\{0,...,m\}$ into sets $S_1,...,S_{m'}$ of size about $c\sqrt{m}$. Then, for every $1\le i\le m'$ and $x\in \{0,1\}^{S_i}$ we determine the most likely value of the restriction of the codeword to the corresponding bits from its noisy version. We then vote, i.e., we guess that every bit in the codeword has the value that the largest number of these reconstructions assign to it. For every $x_1\in\{0,1\}^{S_1}$, $x_2\in\{0,1\}^{S_2}$,...,$x_{m'}\in\{0,1\}^{S_{m'}}$, there is exactly one bit of the code having coordinates in $S_i$ taking on the values specified by $x_i$ for all $i$. Hence the number of bits that the majority of the partial reconstructions get wrong is at most the number of tuples $(x_1,...,x_{m'})$ such that at least half of the $x_i$ are bad. This allow us to get most of the bits right except in the extreme low probability case with unusually many bad $x$'s. That in turn lets us apply the previous list-decoding technique to reach a block error probability of $2^{-2^{\Omega(\sqrt{m}})}$.

\section{Conclusion and open problems}
This paper proves the conjecture that RM codes achieve capacity on all BMS channels, in the classical sense of recovering the codewords down to capacity. It shows that, despite their simple deterministic channel-independent construction taking $\tilde{O}(n)$ time,  RM codes match random codes in universally achieving capacity on BMS channels.

The paper puts forward in particular a new boosting framework for decoding.  Boosting techniques have been developed extensively in supervised learning \cite{schapire}, as well as community detection \cite{abbe_fnt} where weak recovery consist in recovering a coordinate with error probability at most $1/2-\Omega(1)$. However, boosting techniques have not been used extensively in coding theory so far, and weak recovery has also not been considered extensively in this context. Our boosting requires as base case a weak recovery notion for decoding a bit given the other noisy bits, which must be non-trivial but possibly constant (or even slowly increasing to $1/2$). This weak base case benefits from being very general: we show that every code below capacity affords such a base case as long as it is symmetrical (i.e., making the coordinate choice irrelevant). In addition, it is much simpler to establish this base case than a vanishing scaling of $O(\log(m)/\sqrt{m})$ as in \cite{Reeves21}, while a direct round of sunflower-boosting already gives a faster decay to $O(2^{-\sqrt[3]{m}})$. As discussed in Section \ref{critical}, the main challenge then lies in reaching a scaling of $2^{-\Omega(\sqrt{m}\log(m))}$, which the recursive boosting algorithm achieves with the $L_2,L_4$ anti-concentration analysis. Two natural directions are now: 


\begin{itemize}
\item {\bf The decoding complexity:} the main open problem regarding RM codes is now to obtain also an efficient decoding algorithm, ideally down to capacity. This has also been an active area of research over the last decades, with results surveyed in \cite{RM_fnt}. Polar codes have currently the significant advantage of being very efficiently decodable, at the expense of having a more involved (albeit still efficient \cite{Tal13}) and channel-dependent encoding. Further, as shown in \cite{Hassani14,Guruswami15} and with the general strong polarization \cite{strong_polar}, polar codes have a polynomial gap to capacity, and thus approach capacity with an overall polynomial complexity. In fact polar codes have practical relevance and have entered 5G protocols \cite{3gpp}. Their scaling exponent is however not optimal, compared to random codes and to RM codes \cite{Hassani18}, and in some cases RM codes have already proved to have a favorable trade-off \cite{Mondelli14,YA18}.
If one can  obtain an efficient decoding algorithm for RM codes down to capacity, with a quasi-optimal scaling law, the RM code would be astonishing: essentially matching random codes in performance, but efficiently.

Note that the boosting framework could also add to the quest for an efficient decoding algorithm. In fact, the  boosting steps are all quite efficient. It gets to a vanishing bit error probability with $2^{O(\sqrt[3]{m})}$ invocations of the base case, and further down to $e^{-O(\sqrt{m}\log(m))}$ with $e^{O(\sqrt{m}\log(m))}$ invocations of the previous case, and $e^{O(\sqrt{m}\log(m))}$ extra time usage to run. These could be improved. However, two steps require an efficient counter-part to turn the current decoding into an efficient one: (1) the base case, (2) the list decoding argument. Since the base case is now much weaker than a vanishing bit error probability, in fact the weakest possible base probability for our argument to work would be $1/2-2^{-o( \sqrt{m})}$, one may hope for an efficient implementation, but this is currently an open problem. Obtaining an efficient list-decoding for the RM codes is a second interesting open problem. 

\item {\bf Beyond RM codes:} as mentioned several times, our base case is generic and applies to any symmetric code. One could hope to develop the boosting framework in a more general setting than RM codes. Our current version is tailored to evaluations of polynomials on the hypercube, and a first extension would be in such settings.  
To go beyond polynomial evaluation codes, one would need to first understand the behavior of restricted codes.    
\end{itemize}

\section{Setup: boosting framework and sunflowers}
We present first the proof for the binary symmetric channel, BSC$(\epsilon)$, since it captures the essence of the problem, and generalize to other BMS channels in Section \ref{bms}.

Let $m,r\in\mathbb{Z}^+$, $r\le m$, and $\epsilon\in (0,1/2)$. The notations used are as follows,
\begin{enumerate}
\item Draw $f$ uniformly at random in $RM(m,r)$; this is the transmitted codeword;
\item Generate the noise vector $Z\in\mathbb{F}_2^{2^m}$ by independently setting $Z_x$ to $1$ with probability $\epsilon$ and $0$ otherwise for each $x$.
\item Set $\tilde{f}=f+Z$ (component-wise over $\mathbb{F}_2$).
\end{enumerate}

Now, let $L_{m,r,\epsilon}(\tilde{f})$ be the algorithm that finds the most likely value of $f(0^m)$ given the values of $\tilde{f}(x)$ for all $x\ne 0^m$ under the above model, or returns a uniform random value if $0$ and $1$ are equally likely. Note that we remove the observation of $0^m$ since we need to control the dependencies between the petals and this also matches the setting of the base case. 
Let $$P_e(m,r,\epsilon)= \mathbb{P}(L_{m,r,\epsilon}(\tilde{f})\ne f(0^m)).$$ 
\begin{remark}
Observe that for any $\hat{f}\in RM(m,r)$, running $L_{m,r,\epsilon}(\tilde{f}+\hat{f})$ is equivalent to running $L_{m,r,\epsilon}(\tilde{f})$ and then adding $\hat{f}(0^m)$ to its output. So, the probability that $L_{m,r,\epsilon}(\tilde{f})$ returns the correct value depends only on the value of $Z$. 
\end{remark}
\begin{remark}
Providing or not the noisy bit for $0^m$ to the decoder, i.e.,using  the bit-error probability versus the exit function, makes a difference when these probabilities are of constant order, but no longer when they are vanishing (since they differ by a constant factor). Hence having a vanishing bit-error probability is equivalent to having a vanishing exit function.   
\end{remark}

\begin{figure}
    \centering
    \includegraphics[width=.6\textwidth]{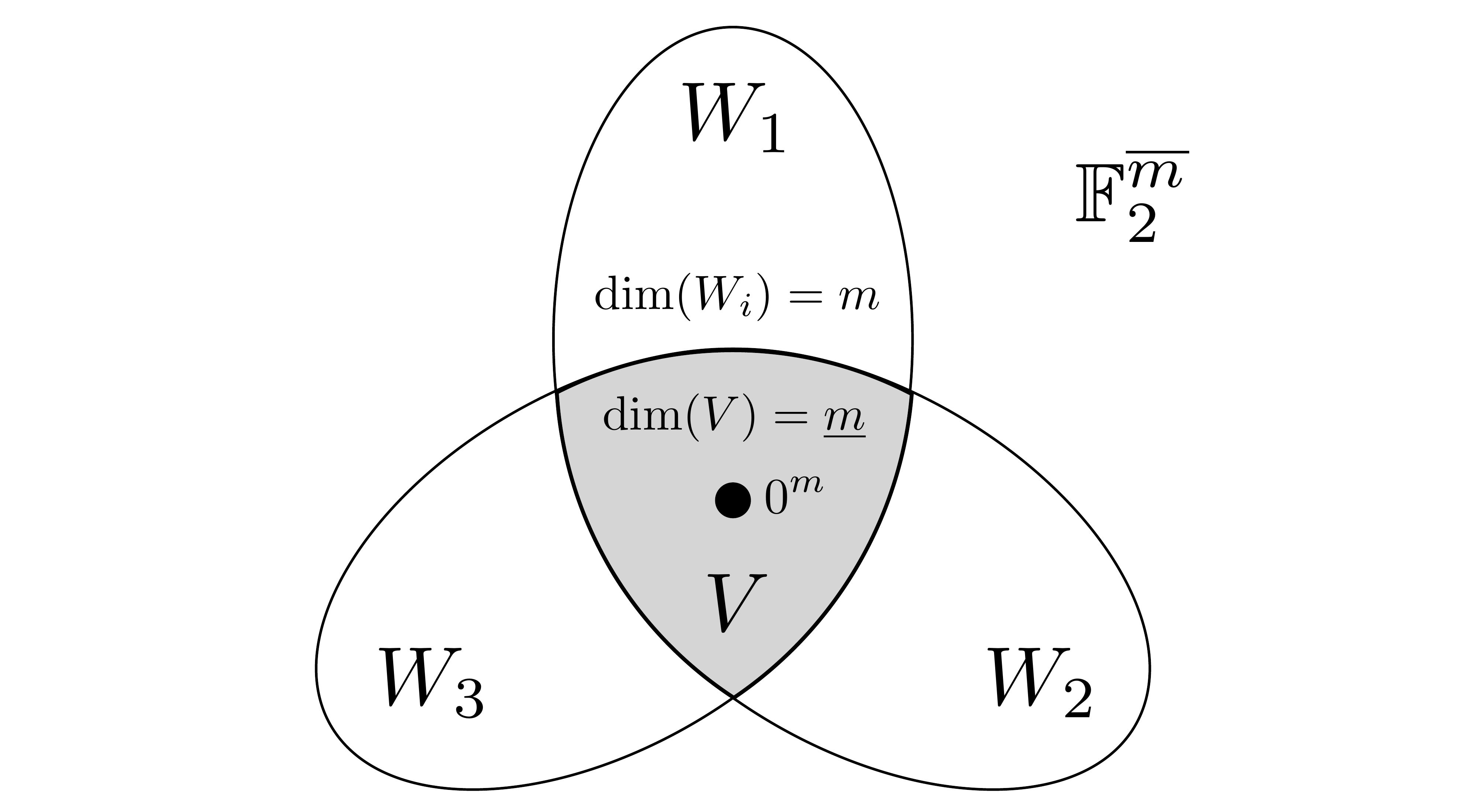}
    \caption{A $(\underline{m},m,\overline{m})$-subspace sunflower illustration with 3 petals $W_1,W_2,W_3$ and kernel $V$.}
    \label{fig:sunflower}
\end{figure}

\begin{definition}
 For any $0\le \underline{m}<m<\overline{m}$ and $b>0$, a $(\underline{m},m,\overline{m})$-subspace sunflower of size $b$ is a collection of $m$-dimensional subspaces $W_1,...,W_b\subseteq\mathbb{F}_2^{\overline{m}}$ such that there exists an $\underline{m}$-dimensional space $V\subseteq\mathbb{F}_2^{\overline{m}}$ such that $V\subseteq W_i$ for all $i$ and $W_i\cap W_j=V$ for all $i\ne j$. The $W_i$ are referred to as the petals of the sunflower and $V$ is referred to as its kernel.
\end{definition}
The terminology `sunflower' is used in \cite{sunflower_1} for sets and in \cite{sunflower_2} for subspaces; we only rely here on this existing terminology for the subspace structure of interest, but not need results related to the `sunflower lemma' \cite{sunflower_1} as we can (efficiently) construct our subspace sunflower.  

Now, consider the following algorithm for converting the ability to recover $f(0^m)$ with high accuracy for some specified $m$ and $r$ to a method of recovering $f(0^{\overline{m}})$ with higher accuracy for some $\overline{m}$ moderately higher than $m$ and the same $r$, relying on additional parameters $\underline{m} \le m$ and $b$ as follows:

{\bf BoostingAlgorithm}$(\tilde{f},\underline{m},m,\overline{m},r,\epsilon,b)$:
\begin{enumerate}

\item Pick an $(\underline{m},m,\overline{m})$-subspace sunflower of size $b$.

\item Run $L_{m,r,\epsilon}$ on the restriction of $\tilde{f}$ to each petal of the sunflower, and return the most common output (flipping a fair coin in case of a tie.).
\end{enumerate}

In order to establish that this is a useful approach, we will need a couple of preliminarly results. First of all, we will need to know that it is possible to find a suitably large subspace sunflower, as shown by the following lemma.

\begin{lemma}\label{sunflowerLem}
Let $\underline{m}<m<\overline{m}$. Then there exists a $(\underline{m},m,\overline{m})$-subspace sunflower of size $2^{\overline{m}+\underline{m}+1-2m}$.
\end{lemma}

\begin{proof}
We select an $\underline{m}$-dimensional kernel $V$ arbitrarily, and construct the petals $W_1,...,W_{2^{\overline{m}+\underline{m}+1-2m}}$ by means of the following greedy algorithm:\begin{enumerate}
\item For $1\le i\le 2^{\overline{m}+\underline{m}+1-2m}$: \begin{enumerate}
\item Set $W_i=V$
\item While $\dim(W_i)<m$:\begin{enumerate}
\item Pick $x\not\in \cup_{i'\le i} (W_{i'}+W_i)$.
\item Set $W_i=W_i+\{0,x\}$.
\end{enumerate}
\end{enumerate}
\end{enumerate}

In order to prove that this works, we need to verify a few properties of this algorithm. First of all, each of the $W_i$ is intially set equal to $V$ and then only changed by adding more elements, so it will always be the case that $V\subseteq W_i$ for all $i$. Secondly, the requirement that $x\not\in \cup_{i'\le i} (W_{i'}+W_i)$ ensures that $x\not\in W_i$ and that for every $i'<i$, $W_{i'}\cap (W_i+x)=\emptyset$. Thus it will always be the case that $W_i\cap W_j=V$ for all $i\ne j$. Finally, observe that every time we get to step 1.b.i, the dimensionality of $(W_{i'}+W_i)$ will be at most $2m-\underline{m}-1$ for all $i'<i$ and at most $m-1$ for $i'=i$. Hence there will always be some $x\in\mathbb{F}_2^{\overline{m}}$ that is not in $\cup_{i'\le i} (W_{i'}+W_i)$. So, this algorithm will never get stuck, which ensures that it finds $W_1,...,W_{2^{\overline{m}+\underline{m}+1-2m}}$ with the desired properties, and constructs the desired sunflower.
\end{proof}

Next, observe that given a subspace sunflower, the restrictions of $Z$ to the petals are mutually independent conditioned on the restriction of $Z$ to the kernel. In order to discuss the implications, we define the following.

\begin{definition}
For each $\underline{m}\le m$ and $z'\in\mathbb{F}_2^{2^{\underline{m}}}$, let $P_{e}(\underline{m},m,r,\epsilon|z')$  be the probability that $L_{m,r,\epsilon}(\tilde{f})\ne f(0^m)$ conditioned on the restriction of $Z$ to $\mathbb{F}_2^{\underline{m}} \times 0^{m-\underline{m}}$ being $z'$. 
\end{definition}
Conditioned on the restriction of $Z$ to $V$ being $z'$, attempting to determine the value of $f(0^{\overline{m}})$ based solely on the restriction of $\tilde{f}$ to $W_i$ will give the wrong answer with probability $P_e(\underline{m},m,r,\epsilon|z')$, and this is independent for all $i$. Therefore, given $\overline{m}$ reasonably larger than $2m-\underline{m}$, it suffices to have $P_e(\underline{m},m,r,\epsilon|z')$ nontrivially less than $1/2$ in order to nearly ensure that we can determine the value of $f(0^{\overline{m}})$ correctly. That brings us to the question of how likely $P_e(\underline{m},m,r,\epsilon|z')$ is to be this small.

\section{Bounding the local error probability}
\subsection{Base case on $P_e(m,r,\epsilon)$: weak decoding}
We start by establishing our base case on $P_e(m,r,\epsilon)$ that only requires symmetry in the code. We state the following lemma directly for general BMS channels and general `symmetric' codes (as defined in the lemma).  
\begin{lemma}\label{baseErrorGen}
Let $\mathcal{P}$ be a BMS channel and let $\{C_i\}_{i \ge 1}$ be a sequence of codes of rate $\{R_i\}_{i \ge 1}$ such that $\limsup_{i\to\infty} R_i< C(\mathcal{P})$ and such that a random codeword in $C_i$ has each coordinate marginally distributed as Bernoulli$(1/2)$.  Let $P_e(C_i,\mathcal{P},j)$ be the probability that the maximum likelihood decoder decodes a component $j$ of a random codeword of $C_i$ incorrectly given the observation of the other noisy components. Assume that $P_e(C_i,\mathcal{P},j)$ is independent of $j$ and denote it by $P_e(C_i,\mathcal{P})$.  
Then there exists $c>0$ such that $P_e(C_i,\mathcal{P})<1/2-c$ for all sufficiently large $i$.
\end{lemma}
\begin{proof}
Let $c'=C(\mathcal{P})-\limsup_{i\to\infty} R_i$. 
For all sufficiently large values of $i$, omitting the index $i$ that is implicit and denoting by $X$ a random codeword of $C_i$, $Y$ its corrupted version, and $n=\dim(C_i)$, we have  
\begin{align}
H(Y)&\le H(Y,X) = H(Y|X) + H(X) \\
&\le n( H(\mathcal{P}) + R) = n ( H(Y_1) -(C(\mathcal{P}) -R)) \\
&\le n (H(Y_1)-c'/2). \label{tighten}
\end{align}
Moreover, defining $Y_{<j}=(Y_1,\ldots, Y_{j-1})$ and $Y_{-j}=(Y_1,\ldots, Y_{j-1}, Y_{j+1},\ldots, Y_n)$, we have 
\begin{align}
H(Y)= \sum_{j=1}^{n} H(Y_j|Y_{<j}),\end{align}
 and there must exist $j$ for which 
 \begin{align}
H(Y_j|Y_{<j}) \le H(Y_1)-c'/2  =  H(Y_j)-c'/2 .
 \end{align}
Furthermore, 
 \begin{align}
 H(X_1)-H(X_1 |Y_2,\ldots,Y_n)&=I(X_1;Y_2,\ldots, Y_n)\\
 &\ge I(Y_1;Y_2,\ldots, Y_n)= I(Y_j;Y_{-j}) \label{condi}\\
 &\ge I(Y_j;Y_{<j}) \ge c'/2,
\end{align}
 where the inequality in \eqref{condi} uses $H(Y_2,\ldots,Y_n|X_1,Y_1)= H(Y_2,\ldots,Y_n|X_1)$. Thus 
 \begin{align}
 H(X_1 |Y_2,\ldots, Y_n)&\le 1-c'/2.
\end{align}
Finally, since for a binary random variable $U$, $1-\max_{u}p_U(u)=1-2^{-H_\infty(U)}\le 1-2^{-H(U)}$, where $H_\infty(U)$ is the min-entropy, we have by Jensen's inequality (to handle the conditional entropy), 
\begin{align}
P_e(C,\mathcal{P})\le 1-2^{-H(X_1 |Y_2,\ldots, Y_n)} 
\end{align}
and we can take any $c<(2^{c'/2}-1)/2$ (or $c<(2^{c'}-1)/2$ by tightening \eqref{tighten}).
\end{proof}

The symmetry condition holds for any doubly transitive codes such as RM codes \cite{Kudekar16STOC}, implying the following. 
\begin{corollary}\label{baseError}
Let $\epsilon\in (0,1/2)$, and $\{m_i\}_{i \ge 1}$ and $\{r_i\}_{i \ge 1}$ be sequences of positive integers such that $r_i \le m_i$ for all $i$ and $\lim_{i\to\infty} m_i=\infty$ and $\limsup_{i\to\infty} {m_i \choose \le r_i}2^{-m_i}< 1-H(\epsilon)$. Then there exists $c>0$ such that $P_e(m_i,r_i,\epsilon)<1/2-c$, for all sufficiently large $i$.
\end{corollary}

\subsection{Fourier analysis of $P_e(\underline{m},m,r,\epsilon | z')$}
From now on, let $k=m-\underline{m}$ and consider $\epsilon, r$ fixed. In order to analyze the behavior of the function $Q_{\underline{m},m}: \mathbb{F}_2^{\underline{m}} \ni z' \mapsto P_e(\underline{m},m,r,\epsilon|z')$, we are going to use Fourier analysis. Given a function $\mathcal{G}: \{0,1\}^{2^{\underline{m}}}\rightarrow\mathbb{R}$, let $\langle \mathcal{G}\rangle=\sum_{z'\in \{0,1\}^{2^{\underline{m}}}}\epsilon^{|z'|}(1-\epsilon)^{2^{\underline{m}}-|z'|}\mathcal{G}(z')$ and given two functions $\mathcal{G},\mathcal{H}: \{0,1\}^{2^{\underline{m}}}\rightarrow\mathbb{R}$, let 
\[ \langle \mathcal{G}, \mathcal{H}\rangle=\sum_{z'\in \{0,1\}^{2^{\underline{m}}}}\epsilon^{|z'|}(1-\epsilon)^{2^{\underline{m}}-|z'|}\mathcal{G}(z')\mathcal{H}(z').\]
Now, for each $S\subseteq\mathbb{F}_2^{\underline{m}}$, let $\mathcal{X}_S:\{0,1\}^{2^{\underline{m}}}\rightarrow\mathbb{R}$ be the function such that
\[\mathcal{X}_S(z')=\left(\frac{\epsilon}{1-\epsilon}\right)^{|S|/2}\left(-\frac{1-\epsilon}{\epsilon}\right)^{|\{x\in S:z'_x=1\}|}\]
for all $z'$. These form the orthonormal Fourier-Walsh basis for the biased measure on the space of functions from $\{0,1\}^{2^{\underline{m}}}$ to $\mathbb{R}$, so 
\[Q_{\underline{m},m}=\sum_{S\subseteq\mathbb{F}_2^{\underline{m}}} \langle Q_{\underline{m},m}, \mathcal{X}_S\rangle\mathcal{X}_S.\]
Also, since $L_{m,r,\epsilon}(\tilde{f})$ is independent of $\tilde{f}(0^m)$, $P_e(\underline{m},m,r,\epsilon|z')$ is independent of $z'_0$ and $\langle Q_{\underline{m},m}, \mathcal{X}_S\rangle=0$ whenever $0^{\underline{m}}\in S$. Now, for each $x\in\mathbb{F}_2^{\underline{m}}$, let $e_m(x)$ be the element of $\mathbb{F}_2^m$ having first $\underline{m}$ coordinates  given by $x$ and remaining coordinates set to $0$. Next, observe that for all $\underline{m}<m$,
\[ Q_{\underline{m},m}=\sum_{S\subseteq\mathbb{F}_2^{\underline{m}}} \langle Q_{m,m}, \mathcal{X}_{e_m(S)}\rangle\mathcal{X}_{S}\]
because for any $S'$ that is not expressible as $\mathcal{X}_{e_m(S)}$ for some $S\subseteq \mathbb{F}_2^{\underline{m}}$, there exists $x\in S'$ such that $Z_x$ is still random conditioned on the restriction of $Z$ to $\mathbb{F}_2^{2^{\underline{m}}}$ being $z'$, so its contribution to $Q_{\underline{m},m}$ cancels itself out. Next, observe that
\begin{align}
\mathbb{P}[P_e(\underline{m},m,r,\epsilon|z')\ge 1/3]
&\le \mathbb{E}[9 P^2_{\underline{m},m,r,\epsilon}(z')]\\
&= 9\langle Q_{\underline{m},m}, Q_{\underline{m},m}\rangle\\
&=9\sum_{S\subseteq \mathbb{F}_2^{\underline{m}}} \langle Q_{m,m},\mathcal{X}_{e_m(S)}\rangle^2
\end{align}
Thus in order to bound the probability that $P_e(\underline{m},m,r,\epsilon|z')\ge 1/3$, it suffices to show that the products $\langle Q_{m,m},\mathcal{X}_{e_m(S)}\rangle$ are sufficiently small.

Next, observe that $Q_{m,m}$ is symmetric under affine transformations. So, if $S'$ is an affine transformation of $S$ (more specifically a linear transformation) then it must be the case that $\langle Q_{m,m},\mathcal{X}_{S}\rangle=\langle Q_{m,m},\mathcal{X}_{S'}\rangle$. Now, for each $S\subseteq \mathbb{F}_2^{m}$, let $\overline{S}$ be the collection of affine transformations of $S$ and $\mathcal{X}_{\overline{S}}=\sum_{S'\in\overline{S}}\mathcal{X}_{S'}$. Also, let $\mathbb{S}$ be a maximal list of subsets of $\mathbb{F}_2^{m}$ that are not equivalent under affine transformations. 
Then, for each $S\subseteq\mathbb{F}_2^{m}$, let $\dim(S)$ denote the dimension of the subspace of $\mathbb{F}_2^{m}$ spanned by $S$. Given $S\subseteq\mathbb{F}_2^{m}$, there exist $\dim(S)$ nonzero linearly independent points $x_1,...,x_{\dim(S)}\in S$. For any linear transformation $\pi$, $\pi(S)$ is contained in $\mathbb{F}_2^{\underline{m}}$ if and only if $\pi(x_i)\in \mathbb{F}_2^{\underline{m}}$ for all $i$. Furthermore, for a random linear transformation $\pi$ and conditioned on any fixed values of $\pi(x_1),...,\pi(x_{i-1})$, the probability distribution of $\pi(x_i)$ is the uniform distribution on the points that are not contained in $Span(\pi(x_1),...,\pi(x_{i-1}))$. Hence $\mathbb{P}[\pi(S)\subseteq\mathbb{F}_2^{\underline{m}}]\le (2^{m-\underline{m}})^{\dim(S)}=2^{-k\cdot \dim(S)}$. Therefore  

\begin{align}
\sum_{S\subseteq \mathbb{F}_2^{\underline{m}}} \langle Q_{m,m},\mathcal{X}_{e_m(S)}\rangle^2
&\le \sum_{S\in\mathbb{S}} |\overline{S}|2^{-k\cdot \dim(S)}\langle Q_{m,m},\mathcal{X}_{S}\rangle^2\\
&= \sum_{S\in\mathbb{S}} \frac{1}{|\overline{S}|}2^{-k\cdot \dim(S)}\langle Q_{m,m}, \mathcal{X}_{\overline{S}}\rangle^2.
\end{align}

\subsection{Vanishing bit-error probability via single sunflower-boosting and $L_2$-norm}
Recall that $k=m-\underline{m}$. 
\begin{lemma}\label{boostLemma1}
Let $0\le \underline{m}< m$ and $r\ge 0$ be integers and $\epsilon\in(0,1/2)$. Then when $z'\in\{0,1\}^{2^{\underline{m}}}$ is selected by independently setting each of its elements to $1$ with probability $\epsilon$ and $0$ otherwise,

\[\mathbb{P}[P_e(\underline{m},m,r,\epsilon|z')\ge P_e(m,r,\epsilon)/2+1/4]\le \frac{2^{2-k}}{(1/2-P_e(m,r,\epsilon))^2}.\]
\end{lemma}


\begin{proof}
First, let $c=1/2-P_e(m,r,\epsilon)$, and observe that 
\begin{align}
 & \mathbb{P}[P_e(\underline{m},m,r,\epsilon|z')\ge 1/2-c/2]\\
 &\le \mathbb{E}[(P_e(\underline{m},m,r,\epsilon|z')-(1/2-c))^2]] (4/c^2)\\ 
 &=\frac{4}{c^2}\sum_{S\subseteq\mathbb{F}_2^{\underline{m}}:S\ne\emptyset} \langle Q_{m,m}, \mathcal{X}_{e_m(S)}\rangle^2\\
&\le \frac{4}{c^2} \sum_{S\in\mathbb{S}:S\ne\emptyset} \frac{1}{|\overline{S}|}2^{-k\cdot \dim(S)}\langle Q_{m,m},\mathcal{X}_{\overline{S}}\rangle^2\\
&= \frac{4}{c^2} \sum_{S\in\mathbb{S}:S\ne\emptyset,0^{\underline{m}}\not\in S} \frac{1}{|\overline{S}|}2^{-k\cdot \dim(S)}\langle Q_{m,m}, \mathcal{X}_{\overline{S}}\rangle^2\\
&\le \frac{4}{c^2} \sum_{S\in\mathbb{S}:S\ne\emptyset,0^{\underline{m}}\not\in S} \frac{1}{|\overline{S}|}2^{-k}\langle Q_{m,m},\mathcal{X}_{\overline{S}}\rangle^2\\
&\le \frac{2^{2-k}}{c^2} \sum_{S\in\mathbb{S}} \frac{1}{|\overline{S}|}\langle Q_{m,m},\mathcal{X}_{\overline{S}}\rangle^2\\
&= \frac{2^{2-k}}{c^2}\mathbb{E}[P_e^2(m,m,r,\epsilon|z')]\\
&\le \frac{2^{2-k}}{c^2}
\end{align}
as desired.
\end{proof}

We thus get the following, which implies the bound from \cite{Reeves21}, i.e., a vanishing bit-error probability, with a faster decay in terms of $m$. 
\begin{lemma} \label{errorBound1}
Let $\epsilon\in (0,1/2)$, and $m_1,m_2,m_3,...$ and $r_1,r_2,...$ be sequences of positive integers such that $\lim_{i\to\infty} m_i=\infty$ and $\limsup_{i\to\infty} {m_i \choose \le r_i}2^{-m_i}< 1-H(\epsilon)$. Then $\lim_{i\to\infty} P_e(m_i,r_i,\epsilon)=O(2^{-\sqrt[3]{m_i}})$.
\end{lemma}

\begin{proof}
First, let $m'_i=m_i-2\lfloor \sqrt[3]{m_i}\rfloor$ for each $i$ and observe that $\limsup_{i\to\infty} {m'_i \choose \le r}2^{-m_i}$ is equal to $\limsup_{i\to\infty} {m_i \choose \le r_i}2^{-m_i}$. Now, let $\underline{m_i}=m_i-3\lfloor\sqrt[3]{m_i}\rfloor$, for each $i$ and consider trying to determine the value of $f(0^{m_i})$ from $\tilde{f}$ using the following boosting algorithm: 
\begin{enumerate}
\item Pick a $(\underline{m_i},m'_i,m_i)$-subspace sunflower of size $2^{m_i+\underline{m_i}-2m'_i+1}$ independently of $\tilde{f}$.

\item Run $L_{m'_i,r_i,\epsilon}$ on the restriction of $\tilde{f}$ to each petal of the sunflower, and return the most common output.
\end{enumerate}

By lemma \ref{sunflowerLem} it is possible to find such a subspace sunflower so this algorithm is possible to carry out. Also, if we let $z'$ be the restriction of $Z$ to the kernel then conditioned on a fixed value of $z'$ each use of $L$ in the final step of the algorithm independently returns a value that differs from $f(0^{m_i})$ with probability $P_e(\underline{m_i},m'_i,r_i,\epsilon|z')$. 
By corollary \ref{baseError}, there exists $c>0$ such that $P_e(m'_i,r_i,\epsilon)<1/2-c$ for all sufficiently large $i$. Then the previous lemma implies that
\[\mathbb{P}[P_e(\underline{m_i},m'_i,r_i,\epsilon|z')\ge 1/2-c/2]\le 2^{2-\lfloor\sqrt[3]{m_i}\rfloor}/c^2\]
for any such $i$. Conditioned on any fixed value of $z'$ such that $P_e(\underline{m_i},m'_i,r_i,\epsilon|z')< 1/2-c/2$, this algorithm returns the wrong answer with probability at most $(4(1/2-c)(1/2+c))^{2^{m_i+\underline{m_1}-2m'_i}}$, so its overall probability of returning a value different from $f(0^{m_i})$ is at most $2^{2-\lfloor\sqrt[3]{m_i}\rfloor}/c^2+(1-4c^2)^{2^{\lfloor\sqrt[3]{m_i}\rfloor}}=O(2^{-\sqrt[3]{m_i}})$. Furthermore, this algorithm never uses the value of $\tilde{f}(0^{m_i})$ and $L$ computes $f(0^m)$ at least as accurately as any other algorithm that ignores the value of $\tilde{f}(0^m)$, so it also has an error rate of $O(2^{-\sqrt[3]{m_i}})$, as desired.
\end{proof}

\subsection{Non-concentration on low-dimensional subsets via $L_4$-norm.}
At this point, we know that $P_e(m_i,r_i,\epsilon)$ goes to $0$ as $i$ increases. However, in order to prove that we can recover $f$ with high probability we will need much tighter bounds on $P_e(m_i,r_i,\epsilon)$, and in order to get those we will need to bound the contribution to the Fourier transform of $Q_{m,m}$ of terms with low dimensional support instead of just arguing that the support of every nonconstant term has dimension at least $1$. Cauchy-Schwarz would not establish nontrivial bounds on the Fourier terms; however, we can bound $\langle Q_{m,m},\mathcal{X}_{\overline{S}}\rangle$ using the average of $\mathcal{X}_{\overline{S}}$ over the size-$P_e(m,r,\epsilon)$ fraction of the input space over which it has the highest value, and that in turn can be bounded using the expected value of $\mathcal{X}^b_{\overline{S}}$ for some $b>2$. More formally, for any $S\subseteq\mathbb{F}_2^{m}$, it must be the case that
\begin{align}
&\langle Q_{m,m},\mathcal{X}_{\overline{S}}\rangle\\
&=\sum_{z'\in \{0,1\}^{2^{m}}}\epsilon^{|z'|}(1-\epsilon)^{2^{m}-|z'|}Q_{m,m}(z')\mathcal{X}_{\overline{S}}(z')\\
&\le \left(  \sum_{z'\in \{0,1\}^{2^{m}}}\epsilon^{|z'|}(1-\epsilon)^{2^{m}-|z'|}Q^{4/3}_{m,m}(z')\right)^{3/4}\sqrt[4]{ \sum_{z'\in \{0,1\}^{2^{m}}}\epsilon^{|z'|}(1-\epsilon)^{2^{m}-|z'|}\mathcal{X}^4_{\overline{S}}(z')}\\
&=\langle Q^{4/3}_{m,m}\rangle^{3/4}\sqrt[4]{\langle \mathcal{X}^4_{\overline{S}}\rangle}\\
&\le \langle Q_{m,m}\rangle^{3/4}\sqrt[4]{\langle\mathcal{X}^4_{\overline{S}}\rangle}\\
&= P^{3/4}_{e}(m,r,\epsilon)\sqrt[4]{\langle\mathcal{X}^4_{\overline{S}}\rangle}.
\end{align}

At this point, we need bounds on $\langle\mathcal{X}^4_{\overline{S}}\rangle$, such as the following.

\begin{lemma}
Let $m$ be a positive integer, $S\subseteq\mathbb{F}_2^m$, and $\epsilon\in(0,1/2)$. Then $\langle\mathcal{X}^4_{\overline{S}}\rangle\le 2^{2dm+8d^2}(1/\epsilon(1-\epsilon))^{2^d}$, where $d=\dim(S)$
\end{lemma}

\begin{proof}
First, observe that
\[\langle\mathcal{X}^4_{\overline{S}}\rangle=\sum_{S_1,S_2,S_3,S_4\in\overline{S}} \langle\mathcal{X}_{S_1}\cdot\mathcal{X}_{S_2}\cdot\mathcal{X}_{S_3}\cdot\mathcal{X}_{S_4}\rangle.\]
If there is any $x$ that is in exactly one of $S_1$, $S_2$, $S_3$, and $S_4$ then $\langle\mathcal{X}_{S_1}\cdot\mathcal{X}_{S_2}\cdot\mathcal{X}_{S_3}\cdot\mathcal{X}_{S_4}\rangle=0$ because the contributions of $z'$ for which $z'_x=1$ cancel out the contributions of $z'$ for which $z'_x=0$. If there is no such $x$, then 
\begin{align}
&\langle\mathcal{X}_{S_1}\cdot\mathcal{X}_{S_2}\cdot\mathcal{X}_{S_3}\cdot\mathcal{X}_{S_4}\rangle \\&\le \max(((1-\epsilon)+\epsilon)^{2|S|},(((1-\epsilon)^2+\epsilon^2)/\sqrt{\epsilon(1-\epsilon)})^{4|S|/3},(((1-\epsilon)^3+\epsilon^3)/(\epsilon(1-\epsilon)))^{|S|})\\&\le (1/\epsilon(1-\epsilon))^{|S|}.
\end{align}
We have that 
$S_1\cup S_2\cup S_3\cup S_4$ must contain $\dim(S_1\cup S_2\cup S_3\cup S_4)$ points that are linearly independent. If there is no $x$ that is in exactly one of these sets, then each of these points must be in at least two of the $S_i$, which implies that $d\ge \dim(S_1\cup S_2\cup S_3\cup S_4)/2$. There are at most $2^{2dm}$ $2d$-dimensional subspaces of $\mathbb{F}_2^m$, and there are at most $2^{2d^2}$ affine transformations of $S$ contained in any such subspace. So, $\langle\mathcal{X}^4_{\overline{S}}\rangle\le 2^{2dm+8d^2}(1/\epsilon(1-\epsilon))^{2^d}$, as desired.
\end{proof}

\subsection{Bounding the probability of a large $P_e(\underline{m},m,r,\epsilon | z' )$}

That finally allows us to bound $\sum_{S\subseteq \mathbb{F}_2^{\underline{m}}} \langle Q_{m,m},\mathcal{X}_{e_m(S)}\rangle^2$ as follows:

\begin{lemma}
Let $0\le \underline{m}< m$ and $r\ge 0$ be integers and $\epsilon\in(0,1/2)$. Then when $z'\in\{0,1\}^{2^{\underline{m}}}$ is selected by independently setting each of its elements to $1$ with probability $\epsilon$ and $0$ otherwise,
\[\mathbb{P}[P_e(\underline{m},m,r,\epsilon|z')\ge 1/3]\le 18P^{5/4}_{e}(m,r,\epsilon) +9P_e(m,r,\epsilon)\left(\frac{4\log(64/\epsilon(1-\epsilon))}{\log(1/P_e(m,r,\epsilon))}\right)^k\]
\end{lemma}

\begin{proof}
For any integer $0\le d_0\le m$, 
\begin{align}
&\sum_{S\subseteq \mathbb{F}_2^{\underline{m}}} \langle Q_{m,m},\mathcal{X}_{e_m(S)}\rangle^2\\
&\le \sum_{S\in\mathbb{S}} \frac{1}{|\overline{S}|}2^{-k\cdot \dim(S)}\langle Q_{m,m},\mathcal{X}_{\overline{S}}\rangle^2\\
&=\sum_{d=0}^m \sum_{S\in\mathbb{S}:\dim(S)=d} \frac{1}{|\overline{S}|}2^{-k d}\langle Q_{m,m},\mathcal{X}_{\overline{S}}\rangle^2\\
&=\sum_{d=0}^{d_0-1} \sum_{S\in\mathbb{S}:\dim(S)=d} \frac{1}{|\overline{S}|}2^{-k d}\langle Q_{m,m},\mathcal{X}_{\overline{S}}\rangle^2
+\sum_{d=d_0}^m \sum_{S\in\mathbb{S}:\dim(S)=d} \frac{1}{|\overline{S}|}2^{-k d}\langle Q_{m,m},\mathcal{X}_{\overline{S}}\rangle^2
\end{align}

Now, observe that for any given $d$, there are at most $2^{2^d}$ equivalence classes of $d$-dimensional subsets of $\mathbb{F}_2^m$ under affine transformations because every such set can be mapped to a subset of $\mathbb{F}_2^d$. On the flip side, for every $d$ there are at least $2^{d(m-d)}$ $d$-dimensional subspaces of $\mathbb{F}_2^m$, so for any set $S$ there are at least $2^{\dim(S)(m-\dim(S))}$ sets that can be generated by applying an affine transformation to $S$. So,
\begin{align}
&\sum_{d=0}^{d_0-1} \sum_{S\in\mathbb{S}:\dim(S)=d} \frac{1}{|\overline{S}|}2^{-kd}\langle Q_{m,m},\mathcal{X}_{\overline{S}}\rangle^2\\
&\le \sum_{d=0}^{d_0-1} \sum_{S\in\mathbb{S}:\dim(S)=d} \frac{1}{|\overline{S}|}\langle Q_{m,m},\mathcal{X}_{\overline{S}}\rangle ^2\\
&\le \sum_{d=0}^{d_0-1} \sum_{S\in\mathbb{S}:\dim(S)=d} \frac{1}{|\overline{S}|}P^{3/2}_{e}(m,r,\epsilon)\sqrt{\langle\mathcal{X}^4_{\overline{S}}\rangle}\\
&\le \sum_{d=0}^{d_0-1} \sum_{S\in\mathbb{S}:\dim(S)=d} 2^{-d(m-d)} P^{3/2}_{e}(m,r,\epsilon)2^{md+4d^2}(1/\epsilon(1-\epsilon))^{2^{d-1}}\\
&\le \sum_{d=0}^{d_0-1} 2^{2^d} 2^{-d(m-d)}P^{3/2}_{e}(m,r,\epsilon)2^{md+4d^2}(1/\epsilon(1-\epsilon))^{2^{d-1}}\\
& =P^{3/2}_{e}(m,r,\epsilon)  \sum_{d=0}^{d_0-1}2^{5d^2} (4/\epsilon(1-\epsilon))^{2^{d-1}}\\
& \le P^{3/2}_{e}(m,r,\epsilon)  \sum_{d=0}^{d_0-1} (64/\epsilon(1-\epsilon))^{2^d}\\
& \le 2 P^{3/2}_{e}(m,r,\epsilon)  (64/\epsilon(1-\epsilon))^{2^{d_0-1}}.
\end{align}
Now, we can bound the contribution of the high dimensional terms by using the fact that the sum of the squares of all the coefficients in the transform of $Q_{m,m}$ is at most $P_e(m,r,\epsilon)$, and then taking the factor of $2^{-kd}$ into account. More precisely,
\begin{align}
&\sum_{d=d_0}^m \sum_{S\in\mathbb{S}:\dim(S)=d} \frac{1}{|\overline{S}|}2^{-k d}\langle Q_{m,m},\mathcal{X}_{\overline{S}}\rangle^2\\
&\le 2^{-k d_0}\sum_{d=d_0}^m \sum_{S\in\mathbb{S}:\dim(S)=d} \frac{1}{|\overline{S}|}\langle Q_{m,m},\mathcal{X}_{\overline{S}}\rangle^2\\
&\le 2^{-k d_0} \sum_{S\in\mathbb{S}} \frac{1}{|\overline{S}|}\langle Q_{m,m},\mathcal{X}_{\overline{S}}\rangle^2\\
&= 2^{-k d_0} \sum_{S\subseteq\mathbb{F}_2^m} \langle Q_{m,m},\mathcal{X}_{S}\rangle^2\\
&= 2^{-k d_0} \langle Q_{m,m}, Q_{m,m}\rangle\\
&\le 2^{-k d_0} \langle Q_{m,m}\rangle \\
&= P_e(m,r,\epsilon) 2^{-k d_0}.
\end{align}

Thus,
\[\sum_{S\subseteq \mathbb{F}_2^{\underline{m}}} \langle Q_{m,m},\mathcal{X}_{e_m(S)}\rangle^2\le 2 P^{3/2}_{e}(m,r,\epsilon)  (64/\epsilon(1-\epsilon))^{2^{d_0-1}}+P_e(m,r,\epsilon) 2^{-k d_0}.\]

In particular, if we set $d_0=\max\left(\lceil\log_2\left(\frac{\log(1/P_e(m,r,\epsilon))}{4\log(64/\epsilon(1-\epsilon))}\right)\rceil,0\right)$ we get that
\[\sum_{S\subseteq \mathbb{F}_2^{\underline{m}}} \langle Q_{m,m},\mathcal{X}_{e_m(S)}\rangle^2\le 2P^{5/4}_{e}(m,r,\epsilon) +P_e(m,r,\epsilon)\left(\frac{4\log(64/\epsilon(1-\epsilon))}{\log(1/P_e(m,r,\epsilon))}\right)^k\]
That in turn implies that
\begin{align}
&\mathbb{P}[P_e(\underline{m},m,r,\epsilon|z')\ge 1/3]\\
&\le \mathbb{E}[9 P^2_{\underline{m},m,r,\epsilon}(z')]\\
&= 9\langle Q_{\underline{m},m}, Q_{\underline{m},m}\rangle\\
&=9\sum_{S\subseteq \mathbb{F}_2^{\underline{m}}} \langle Q_{m,m},\mathcal{X}_{e_m(S)}\rangle^2\\
&\le 18P^{5/4}_{e}(m,r,\epsilon) +9P_e(m,r,\epsilon)\left(\frac{4\log(64/\epsilon(1-\epsilon))}{\log(1/P_e(m,r,\epsilon))}\right)^k
\end{align}
as desired.
\end{proof}

\subsection{Bounds on $P_e(m,r,\epsilon)$ via recursive sunflower-boosting: reaching the critical regime}
At this point, we can finally start recursively constructing bounds on $P_e$ using the following:
\begin{lemma}\label{boostStepLem}
Let $m$, $\overline{m}$, $k$, and $r$ be positive integers such that $k\le \min(m, \overline{m}-m)$ and $\epsilon\in(0,1/2)$. Then 
\[P_e(\overline{m},r,\epsilon)\le 18P^{5/4}_{e}(m,r,\epsilon) +9P_e(m,r,\epsilon)\left(\frac{4\log(64/\epsilon(1-\epsilon))}{\log(1/P_e(m,r,\epsilon))}\right)^k+(8/9)^{2^{\overline{m}-m-k}}.\]
\end{lemma}

\begin{proof}
Let $\underline{m}=m-k$, and consider decoding the value of $f(0^{\overline{m}})$ from $\tilde{f}$ using the following boosting algorithm: 

\begin{enumerate}
\item Pick an $(\underline{m},m,\overline{m})$-subspace sunflower of size $2^{\overline{m}+\underline{m}-2m+1}$ independently of $\tilde{f}$.

\item Run $L_{m,r,\epsilon}$ on the restriction of $\tilde{f}$ to each petal, and return the most common output.
\end{enumerate}

By lemma \ref{sunflowerLem} it is possible to find such a subspace sunflower so this algorithm is possible to carry out. Also, if we let $z'$ be the restriction of $Z$ to the kernel then conditioned on a fixed value of $z'$ each use of $L$ in the final step of the algorithm independently returns a value that differs from $f(0^{\overline{m}})$ with probability $P_e(\underline{m},m,r,\epsilon|z')$.
\[\mathbb{P}[P_e(\underline{m},m,r,\epsilon|z')\ge 1/3]\le 18P^{5/4}_{e}(m,r,\epsilon) +9P_e(m,r,\epsilon)\left(\frac{4\log(64/\epsilon(1-\epsilon))}{\log(1/P_e(m,r,\epsilon))}\right)^k\]
by the previous lemma, and conditioned on any fixed value of $z'$ such that $P_e(\underline{m},m,r,\epsilon|z')< 1/3$, this algorithm returns the wrong answer with a probability of at most $(8/9)^{2^{\overline{m}+\underline{m}-2m}}=(8/9)^{2^{\overline{m}-m-k}}$. Thus this algorithm gets $f(0^{\overline{m}})$ wrong with a probability of at most
\[18P^{5/4}_{e}(m,r,\epsilon) +9P_e(m,r,\epsilon)\left(\frac{4\log(64/\epsilon(1-\epsilon))}{\log(1/P_e(m,r,\epsilon))}\right)^k+(8/9)^{2^{\overline{m}-m-k}}.\]
Nothing in this algorithm used the value of $\tilde{f}(0^{\overline{m}})$, and $L_{\overline{m},r,\epsilon}(\tilde{f})$ computes $f(0^{\overline{m}})$ from $\tilde{f}$ with accuracy at least as high as that attained by any other algorithm that ignores the value of $\tilde{f}(0^{\overline{m}})$, so it also has an error rate of at most $$18P^{5/4}_{e}(m,r,\epsilon) +9P_e(m,r,\epsilon)\left(\frac{4\log(64/\epsilon(1-\epsilon))}{\log(1/P_e(m,r,\epsilon))}\right)^k+(8/9)^{2^{\overline{m}-m-k}},$$ as desired.
\end{proof}

This allows us to convert bounds on $P_e(m,r,\epsilon)$ into tighter bounds on $P_e(\overline{m},r,\epsilon)$ for $\overline{m}$ slightly greater than $m$. Our plan is to now use this to prove increasingly tighter bounds, starting with the bound established by lemma \ref{errorBound1}. That yields the following result.

\begin{lemma} \label{bitError}
Let $\epsilon\in (0,1/2)$, and $m_1,m_2,m_3,...$ and $r_1,r_2,...$ be sequences of positive integers such that $\lim_{i\to\infty} m_i=\infty$ and $\limsup_{i\to\infty} {m_i \choose \le r_i}2^{-m_i}< 1-H(\epsilon)$. Then there exists $c>0$ such that $P_e(m_i,r_i,\epsilon)=O(m_i^{-c \sqrt{m_i}})$.
\end{lemma}

\begin{proof}
First, observe that for any sufficiently small constant $c'$, it will be the case that $$\limsup_{i\to\infty} {m_i-c'\sqrt{m_i} \choose \le r}2^{-m_i}$$ is also strictly less than $1-H(\epsilon)$. So, choose such a $c'$ and let $m'_i=m_i-\lfloor c'\sqrt{m_i}\rfloor$ for each $i$, and recall that $P_e(m'_i,r_e,\epsilon)=O(2^{-\sqrt[3]{m_i}})$ by lemma \ref{errorBound1}. That means that that there exists $i_0$ such that for all $i\ge i_0$, $m_i\ge 2$ and
\[P_e(m_i-\lfloor c'\sqrt{m_i}\rfloor,r_i,\epsilon)\le \min( e^{-108\ln(64/\epsilon(1-\epsilon))\sqrt[4]{m_i}},2^{-2 \log_2^2(m_i)}/54^4).\]
Hence for any $i\ge i_0$ and $m''\ge m_i-\lfloor c'\sqrt{m_i}\rfloor$, it must be the case that
\begin{align}
&P_e(m''+2\lceil \log_2{m_i}\rceil,r_i,\epsilon)\\
&\le 18P^{5/4}_{e}(m'',r_i,\epsilon) +9P_e(m'',r_i,\epsilon)\left(\frac{4\log(64/\epsilon(1-\epsilon))}{\log(1/P_e(m'',r_i,\epsilon))}\right)^{\log_2(m_i)}+(8/9)^{m_i}\\
&\le 2^{-\log_2^2(m_i)/2} P_e(m'',r_i,\epsilon)/3+m_i^{-\log_2(m_i)/4}P_e(m'',r_i,\epsilon)/3+(8/9)^{m_i}\\
&\le \max(2^{-\log_2^2(m_i)/4} P_e(m'',r_i,\epsilon),3(8/9)^{m_i}).
\end{align}
Applying this $\left\lfloor \frac{\lfloor c'\sqrt{m_i}\rfloor}{2\lceil\log_2 m_i\rceil}\right\rfloor$ times implies that for all $i\ge i_1$,
\[P_e(m_i,r_i,\epsilon)\le \max\left(2^{-c' \log_2(m_i)\sqrt{m_i}/16},3(8/9)^{m_i}\right).\]
This gives us the desired conclusion for $c=c'/16$.
\end{proof}

\section{Recovering the codeword}\label{recovering}
At this point, we are finally ready to go from being able to recover any given bit of the code with high probability to being able to completely recover the codeword with high probability. 

\subsection{List-decoding and weight enumerator}
In order to do that, we will start by showing that it is very unlikely that there is a codeword that is close to the true codeword and more plausible than it in the following sense.

\begin{lemma}\label{listDecodeLem}
Let $c>0$, $\epsilon\in (0,1/2)$, and $m_1,m_2,...$ and $r_1,r_2,...$ be sequences of positive integers such that $\lim_{i\to\infty} m_i=\infty$ and $|m_i-2r_i|=O(\sqrt{m_i})$. Then in the limit as $i\to\infty$, if $f$ is randomly drawn from $RM(m_i,r_i)$ and $\tilde{f}$ is a noisy version of $f$ in which each bit is independently flipped with probability $\epsilon$ then with probability $1-O(2^{-2^{m/3}})$ there is no $f^\star\in RM(m_i,r_i)$ different from $f$ such that $f^\star$ disagrees with $f$ on at most $2^{m-c\sqrt{m}\log_2(m)}$ elements and $f^\star$ agrees with $\tilde{f}$ on at least as many elements as $f$ does.
\end{lemma}


\begin{proof}
First, observe that from page 17 in the proof of Theorem 1.2 from \cite{Sberlo18}, the number of elements of $RM(m,r)$ that disagree with $f$ on at most $2^{m-\l}$ elements is at most 
\[2^{\sum_{j=\l}^r 17m(j-1)(j+2)+17(j+2){m-(j-1) \choose \le r-(j-1)}}\]
 for every $\l \in\mathbb{Z}^+$.
Given $f^\star\in RM(m,r)$, the probability that $f^\star$ has at least as much agreement with $\tilde{f}$ as $f$ does is at most $(4\epsilon(1-\epsilon))^{|\{x:f(x)\ne f^\star(x)\}|/2}$, as we need to flip a number of bits at least equal to half the distance. So, the probability that there is a codeword other than $f$ which is within $2^{m-c\sqrt{m}\log_2(m)}$ elements of $f$ and which agrees with $\tilde{f}$ at least as much as $f$ does is at most:
\begin{align}
&\sum_{f^\star\in RM(m,r): 0<|\{x:f(x)\ne f^\star(x)\}|\le 2^{m-\lfloor c\sqrt{m}\log_2(m)\rfloor }} (4\epsilon(1-\epsilon))^{|\{x:f(x)\ne f^\star(x)\}|/2}\\
&=\sum_{\l= \lfloor c\sqrt{m}\log_2(m)\rfloor}^r\sum_{f^\star\in RM(m,r): 2^{m-\l-1}<|\{x:f(x)\ne f^\star(x)\}|\le 2^{m-\l}} (4\epsilon(1-\epsilon))^{|\{x:f(x)\ne f^\star(x)\}|/2}\\
&\le \sum_{\l= \lfloor c\sqrt{m}\log_2(m)\rfloor}^r\sum_{f^\star\in RM(m,r): 2^{m-\l-1}<|\{x:f(x)\ne f^\star(x)\}|\le 2^{m-\l}} (4\epsilon(1-\epsilon))^{2^{m-\l-2}}\\
&\le \sum_{\l= \lfloor c\sqrt{m}\log_2(m)\rfloor}^r (4\epsilon(1-\epsilon))^{2^{m-\l-2}}2^{\sum_{j=\l}^r 17m(j-1)(j+2)+17(j+2){m-(j-1) \choose \le r-(j-1)}}.
\end{align}

Now, observe that $|m-2r|=O(\sqrt{m})$, so there exists $i_0$ such that $|m-2r|\le c\sqrt{m}\log_2(m)/2-2$ whenever $i\ge i_0$. Given such an $i$, observe that for any $j\ge \l\ge \lfloor c\sqrt{m}\log_2(m)\rfloor$, a Chernoff bound implies that
\begin{align}
{m-(j-1) \choose \le r-(j-1)}&\le 2^{m-(j-1)}e^{-\left(\frac{m-r}{m-j+1}-\frac{1}{2}\right)^2(m-j+1)}\\
&=2^{m-(j-1)}e^{-\frac{(m-2r+j-1)^2}{4(m-j+1)}}\\
&\le 2^{m-(j-1)}e^{-\frac{(c\sqrt{m}\log_2(m)/2)^2}{4m}}\\
&= 2^{m-(j-1)}e^{-\frac{c^2\log_2^2(m)}{16}}.
\end{align}
Hence for any sufficiently large value of $i$, the probability that there is a codeword other than $f$ which is within $2^{m-c\sqrt{m}\log_2(m)}$ elements of $f$ and that agrees with $\tilde{f}$ on at least as many elements as $f$ is at most:
\begin{align}
&\sum_{\l= \lfloor c\sqrt{m}\log_2(m)\rfloor}^r (4\epsilon(1-\epsilon))^{2^{m-\l-2}}2^{\sum_{j=\l}^r 17m(j-1)(j+2)+17(j+2)2^{m-(j-1)}e^{-\frac{c^2\log_2^2(m)}{16}}}\\
&\le \sum_{\l= \lfloor c\sqrt{m}\log_2(m)\rfloor}^r (4\epsilon(1-\epsilon))^{2^{m-\l-2}}2^{17mr(r-1)(r+2)+17r(r+2)2^{m-(\l-1)}e^{-\frac{c^2\log_2^2(m)}{16}}}\\
&= 2^{17mr(r-1)(r+2)}\sum_{\l= \lfloor c\sqrt{m}\log_2(m)\rfloor}^r 2^{\left(136 r(r+2)e^{-\frac{c^2\log_2^2(m)}{16}}-\log_2(1/4\epsilon(1-\epsilon))\right) 2^{m-\l-2}}.
\end{align}
Now, observe that $136r(r+2))e^{-\frac{c^2\log_2^2(m)}{16}}=O(1)$, so whenever $i$ is sufficiently large this will be less than
\begin{align}
& 2^{17mr(r-1)(r+2)}\sum_{\l= \lfloor c\sqrt{m}\log_2(m)\rfloor}^r 2^{-\log_2(1/4\epsilon(1-\epsilon)) 2^{m-\l-3}}\\
&\le r 2^{17mr(r-1)(r+2)-\log_2(1/4\epsilon(1-\epsilon)) 2^{m-r-3}}=O(2^{-\sqrt[3]{n}}).
\end{align}
Thus no such $f^\star$ exists with probability $1-O(2^{-\sqrt[3]{n}})$, as desired.
\end{proof}

That suggests the following algorithm for recovering the codeword.
\begin{algorithm}
{\bf $RM\_reconstruction\_algorithm(m,r,\epsilon,\tilde{f},c):$}\begin{enumerate}
\item For each $x$, let $\hat{f}(x)$ be the most likely value of $f(x)$ given $\tilde{f}$, or a uniform random value if $0$ and $1$ are equally likely.

\item Let $\ell$ be a list of every element of $RM(m,r)$ that disagrees with $\hat{f}$ on at most $2^{m-c\sqrt{m}\log_2(m)/2}/2$ elements.

\item Return the element of $\ell$ that agrees with $\tilde{f}$ on the largest number of elements, breaking ties uniformly at random.

\end{enumerate}
\end{algorithm}

We claim that this completely recovers the original codeword with high probability in the following sense
\begin{theorem} \label{recoveryTheorem}
Let $\epsilon\in (0,1/2)$, and $m_1,m_2,...$ and $r_1,r_2,...$ be sequences of positive integers such that $\lim_{i\to\infty} m_i=\infty$ and $\limsup_{i\to\infty} {m_i \choose \le r_i}2^{-m_i}< 1-H(\epsilon)$. Also let $c$ be the constant stated to exist by lemma \ref{bitError}. Then in the limit as $i\to\infty$, if $f$ is randomly drawn from $RM(m,r)$ and $\tilde{f}$ is a noisy version of $f$ in which each bit is independently flipped with probability $\epsilon$ then $RM\_reconstruction\_algorithm(m,r,\epsilon,\tilde{f},c)$ returns $f$ with probability $1-O(2^{-c\sqrt{m}\log_2(m)/2})$.
\end{theorem}

\begin{proof}
First, note that by lemma \ref{bitError} and symmetry between different values of $x$, $\mathbb{P}[\hat{f}(x)\ne f(x)]=O(m^{-c\sqrt{m}})$ for all $x$. So, $f\in \ell$ with probability $1-O(2^{-c\sqrt{m}\log_2(m)/2})$, in which case every element of $\ell$ will disagree with $f$ on at most $2^{m-c\sqrt{m}\log_2(m)/2}$ elements. By the previous lemma, the probability that this happens and there is an element of $\ell$ that agrees with $\tilde{f}$ on at least as many elements as $f$ does is $1-O(2^{-2^{m/3}})$ (noting that $\ell\subseteq RM(m_i,\max(r_i,m_i/2))$ so we do not need to worry about being unable to apply the lemma in the case where $r_i$ would be too small). This algorithm recovers thus $f$ with probability $1-O(2^{-c\sqrt{m}\log_2(m)/2})$, as desired.
\end{proof}

\subsection{Tightening the block error probability via grid-boosting}

At this point, we know that we can recover the code with high probability, but our bound on the error probability goes to $0$ fairly slowly. We will next prove better bounds on the error probability, and in order to do that, we will consider the following algorithm.

\begin{algorithm}
{\bf $RM\_reconstruction\_algorithm2(m,r,\epsilon,\tilde{f},c,c'):$}\begin{enumerate}
\item Let $m'=\lceil \sqrt{m}/c\rceil$. 

\item Divide the indices $1,...,m$ into sets $S_1,...,S_{m'}$ with sizes that are as close as possible to being equal.

\item For each $1\le i\le m'$ and $x\in\{0,1\}^{S_i}$, set \\$f^{(i,x)}=RM\_reconstruction\_algorithm(m,r,\epsilon,\tilde{f}_{S_i,x},c')$, where $\tilde{f}_{S_i,x}$ is the restriction of $\tilde{f}$ to points thats coordinates indexed in $S_i$ take on the values specified by $x$.

\item For each $x\in\{0,1\}^m$, set $\widehat{f}(x)$ to the most common value of $f^{(i,x_{S_i})}(x)$ over $1\le i\le m'$, where $x_{S_i}$ is the substring of $x$ with indices in $S_i$.

\item Let $\ell$ be a list of every element of $RM(m,r)$ that disagrees with $\hat{f}$ on at most $2^{m-\sqrt{m}\log_2(m)}/2$ elements.

\item Return the element of $\ell$ that agrees with $\tilde{f}$ on the largest number of elements.
\end{enumerate}
\end{algorithm}

 We claim that this algorithm recovers the code with very high probability in the following sense.

 \begin{theorem}
Let $\epsilon\in (0,1/2)$, and $m_1,m_2,...$ and $r_1,r_2,...$ be sequences of positive integers such that $\lim_{i\to\infty} m_i=\infty$ and $\limsup_{i\to\infty} {m_i \choose \le r_i}2^{-m_i}< 1-H(\epsilon)$. There exist $c,c',c''>0$ such that in the limit as $i\to\infty$, if $f$ is randomly drawn from $RM(m,r)$ and $\tilde{f}$ is a noisy version of $f$ in which each bit is independently flipped with probability $\epsilon$ then $RM\_reconstruction\_algorithm2(m_i,r_i,\epsilon,\tilde{f},c,c')$ returns $f$ with probability $1-O(2^{-2^{c''\sqrt{m}}})$.
\end{theorem}

\begin{proof}
First, observe that for sufficiently small values of $c$, $\limsup_{i\to\infty} {m_i-c\sqrt{m_i} \choose \le r}2^{-m_i}<1-H(\epsilon)$ as well, and pick some such $c$. When this algorithm is run, it will always be the case that $c\sqrt{m}-c^2-1\le |S_j|\le c\sqrt{m}+1$, so for an arbitrary series $j_1,j_2,\ldots$ it will be the case that $\limsup_{i\to\infty} {m_i-|S_{j_i}| \choose \le r}2^{-m_i}<1-H(\epsilon)$. Now, let $c'$ be the constant stated to exist by lemma \ref{bitError} for $(m'_1,r'_1),(m'_2,r'_2),...$ listing the values of $(m_i-|S_j|,r_i)$ for all possible values of $i$ and $j$. 

Next, for each $1\le j\le m'$ and $x\in \{0,1\}^{S_j}$, let $B_{j,x}$ be the indicator variable that is $0$ if $f^{(j,x)}$ is equal to the appropriate restriction of $f$ and $1$ if it is not. By the previous theorem, $\mathbb{P}[B_{j,x}=1]=O(2^{-c'\sqrt{m}\log_2(m)/2})$ for any choice of $j$ and $x$. Furthermore, $B_{j,0^{S_j}},...,B_{j,1^{S_j}}$ are independent for any $j$ because the iterations of $RM\_reconstruction\_algorithm$ used to generate them were run on restrictions of $\tilde{f}$ to nonnoverlapping sets. As a result, 
\[\mathbb{P}\left[\sum_{x\in\{0,1\}^{S_j}} B_{j,x}\ge 2^{c\sqrt{m}/2}\right]\le 2^{|S_j|\cdot c\sqrt{m}/2}2^{\left(O(1)-c'\sqrt{m}\log_2(m)/2\right)2^{c\sqrt{m}/2}}=O(2^{-2^{c\sqrt{m}/2}}/m). \] If $\sum_{x\in\{0,1\}^{S_j}} B_{j,x}< 2^{c\sqrt{m}/2}$ for all $j$ then
\begin{align}
    &|\{x:\widehat{x}\ne f(x)\}|\\
    &\le\sum_{T\subseteq\{1,...,m'\}:|T|\ge m'/2} \prod_{j\in T} 2^{c\sqrt{m}/2}\prod_{j\not\in T} 2^{|S_j|}\\
    &\le \sum_{T\subseteq\{1,...,m'\}:|T|\ge m'/2} 2^{|T|\cdot c\sqrt{m}/2}2^{(m'-|T|)(c\sqrt{m}+1)}\\
    &\le \sum_{T\subseteq\{1,...,m'\}:|T|\ge m'/2} 2^{(3/4)cm'\sqrt{m}+m'/2}\\
    &\le 2^{(3/4)cm'\sqrt{m}+3m'/2}\\
    &\le 2^{(3/4)m+(3/4)c\sqrt{m}+(3/2)(\sqrt{m}/c+1)}\\
    &=o(2^{m-\sqrt{m}\log_2(m)}).
\end{align}
That in turn means that the probability that $\sum_{x\in\{0,1\}^{S_j}} B_{j,x}< 2^{c\sqrt{m}/2}$ for all $j$ and\\ $RM\_reconstruction\_algorithm2(m_i,r_i,\epsilon,\tilde{f},c,c')$ returns something other than $f$ is $O(2^{-\sqrt[3]{n}})$. So, this algorithm returns $f$ with probability $1-O(m' 2^{-2^{c\sqrt{m}/2}}/m+2^{-\sqrt[3]{n}})=1-O(2^{-2^{c\sqrt{m}/2}})$, as desired. 
\end{proof}

\section{Strong converse}\label{converse}
The only part of the boosting framework that requires the rate of the code to be below the capacity of the channel is the entropy argument establishing a nontrivial base case. So, for any code with rate greater than the capacity of the channel, it must be impossible to recover a bit from the noisy versions of the other bits with accuracy nontrivially greater than $1/2$ because otherwise our argument would imply that we could recover codewords with high probability for some codes with rate greater than the capacity of the channel. The most obvious usage of this argument would merely rule out the ability to recover a bit with accuracy $1/2+\Omega(1)$, but a more careful application of it can set considerably tighter bounds on the accuracy with which we can recover bits. To demonstrate that, consider the following variant of lemma \ref{errorBound1}.

\begin{lemma}
Let $\epsilon\in (0,1/2)$, $c>0$, and $m_1,m_2,m_3,...$ and $r_1,r_2,...$ be sequences of positive integers such that $\lim_{i\to\infty} m_i=\infty$ and $P_e(m_i,r_i,\epsilon)\le 1/2-2^{-c\sqrt{m_i}/3}$ for all $i$. Also, let $m'_i=m_i+2 \lfloor c\sqrt{m_i}\rfloor$ for each $i$. Then $P_e(m'_i,r_i,\epsilon)=O(2^{-c\sqrt{m_i}/3})$.
\end{lemma}

\begin{proof}
First, let $\underline{m_i}=m_i-\lfloor c\sqrt{m_i}\rfloor$, for each $i$ and consider trying to determine the value of $f(0^{m_i})$ from $\tilde{f}$ using the following boosting algorithm: \begin{enumerate}
\item Pick a random $\underline{m_i}$-dimensional subspace of $\mathbb{F_2}^{m'_i}$, $V$.

\item Pick $m_i$-dimensional subspaces $W_1,...,W_{2^{m'_i+\underline{m_i}-2m_i+1}}$  of $\mathbb{F_2}^{m_i}$ independently of $\tilde{f}$ in such a way that $V\subseteq W_j$ for all $j$ and $W_j\cap W_{j'}=V$ for all $j'\ne j$.

\item For each $j$, run $L_{m_i,r_i,\epsilon}$ on the restriction of $\tilde{f}$ to $W_j$, and return the most common output.
\end{enumerate}

By lemma $1$ it is possible to find such a collection of subspaces of $\mathbb{F}_2^{m'_i}$ so this algorithm is possible to carry out. Also, if we let $z'$ be the restriction of $Z$ to $V$ then conditioned on a fixed value of $z'$ each use of $L$ in the final step of the algorithm independently returns a value that differs from $f(0^{m'_i})$ with probability $P_e(\underline{m_i},m_i,r_i,\epsilon|z')$. Then lemma \ref{boostLemma1} implies that

\[\mathbb{P}[P_e(\underline{m_i},m_i,r_i,\epsilon|z')\ge 1/2-2^{-c\sqrt{m_i}/3}/2]\le \frac{2^{2-\lfloor c\sqrt{m_i}\rfloor }}{2^{-2c\sqrt{m_i}/3}}\]
for any $i$. Conditioned on any fixed value of $z'$ such that $P_e(\underline{m_i},m_i,r_i,\epsilon|z')< 1/2-2^{-c\sqrt{m_i}/3}/2$, this algorithm returns the wrong answer with probability at most $(4(1/2-2^{-c\sqrt{m_i}/3}/2)(1/2+2^{-c\sqrt{m_i}/3}/2))^{2^{m'_i+\underline{m_1}-2m_i}}$, so its overall probability of returning a value different from $f(0^{m'_i})$ is at most $2^{2+2c\sqrt{m_i}/3-\lfloor c\sqrt{m_i}\rfloor}+(1-2^{-2c\sqrt{m_i}/3})^{2^{\lfloor c\sqrt{m_i}\rfloor}}=O(2^{-c\sqrt{m_i}/3})$. Furthermore, this algorithm never uses the value of $\tilde{f}(0^{m'_i})$ and $L$ computes $f(0^m)$ at least as accurately as any other algorithm that ignores the value of $\tilde{f}(0^m)$, so it also has an error rate of $O(2^{-c\sqrt{m_i}/3})$, as desired.
\end{proof}

At this point, we can use the boosting argument from before to show that for any $c'$ if $P_e(m'i,r_i,\epsilon)=O(2^{-c\sqrt{m_i}/3})$ then we can recover a codeword drawn from $RM(m_i+c'\sqrt{m_i},r_i)$ with probability $1-o(1)$ after each of its bits is flipped with probability $\epsilon$. More formally, we can prove the following variants of Lemma \ref{bitError} and Theorem \ref{recoveryTheorem}. 

\begin{lemma} \label{BMSBoostLem}
Let $\epsilon\in (0,1/2)$, $c>0$, and $m_1,m_2,m_3,...$ and $r_1,r_2,...$ be sequences of positive integers such that $\lim_{i\to\infty} m_i=\infty$ and $P_e(m_i,r_i,\epsilon)\le 1/2-2^{-c\sqrt{m_i}/3}$ for all $i$. Then there exists $c'>0$ such that $P_e(m_i+3\lfloor c\sqrt{m_i}\rfloor ,r_i,\epsilon)=O(m_i^{-c' \sqrt{m_i}})$.
\end{lemma}

\begin{proof}
We know that $P_e(m_i+2\lfloor c\sqrt{m_i}\rfloor,r_e,\epsilon)=O(2^{-c\sqrt{m_i}/3})$ by the previous lemma. That means that that there exists $i_0$ such that for all $i\ge i_0$, $m_i\ge 2$ and
\[P_e(m_i+2\lfloor c\sqrt{m_i}\rfloor,r_i,\epsilon)\le \min( e^{-108\ln(64/\epsilon(1-\epsilon))\sqrt[4]{m_i}},2^{-2 \log_2^2(m_i)}/54^4).\]
Thus for any $i\ge i_0$ and $m''\ge m_i+2\lfloor c\sqrt{m_i}\rfloor$, it must be the case by lemma \ref{boostStepLem} that
\begin{align}
&P_e(m''+2\lceil \log_2{m_i}\rceil,r_i,\epsilon)\\
&\le 18P^{5/4}_{e}(m'',r_i,\epsilon) +9P_e(m'',r_i,\epsilon)\left(\frac{4\log(64/\epsilon(1-\epsilon))}{\log(1/P_e(m'',r_i,\epsilon))}\right)^{\log_2(m_i)}+(8/9)^{m_i}\\
&\le 2^{-\log_2^2(m_i)/2} P_e(m'',r_i,\epsilon)/3+m_i^{-\log_2(m_i)/4}P_e(m'',r_i,\epsilon)/3+(8/9)^{m_i}\\
&\le \max(2^{-\log_2^2(m_i)/4} P_e(m'',r_i,\epsilon),3(8/9)^{m_i}).
\end{align}
Repeated application of this implies that for all $i\ge i_1$,
\[P_e(m_i+3\lfloor c\sqrt{m_i}\rfloor,r_i,\epsilon)\le \max\left(2^{-c \log_2(m_i)\sqrt{m_i}/16},3(8/9)^{m_i}\right).\]
This gives us the desired conclusion for $c'=c/16$.
\end{proof}

\begin{lemma}
Let $\epsilon\in (0,1/2)$, $c>0$, and $m_1,m_2,...$ and $r_1,r_2,...$ be sequences of positive integers such that $\lim_{i\to\infty} m_i=\infty$, $|r_i-m_i/2|=O(\sqrt{m_i})$, and $P_e(m_i,r_i,\epsilon)\le 1/2-2^{-c\sqrt{m_i}/3}$ for all $i$. Also let $c'$ be the constant stated to exist by the previous lemma. Then in the limit as $i\to\infty$, if $f$ is randomly drawn from $RM(m+3\lfloor c\sqrt{m}\rfloor,r)$ and $\tilde{f}$ is a noisy version of $f$ in which each bit is independently flipped with probability $\epsilon$ then $RM\_reconstruction\_algorithm(m_i+3\lfloor c\sqrt{m_i}\rfloor,r_i,\epsilon,\tilde{f},c')$ returns $f$ with probability $1-O(2^{-c'\sqrt{m}\log_2(m)/2})$.
\end{lemma}

\begin{proof}
First, note that by the previous lemma and symmetry between different values of $x$, $\mathbb{P}[\hat{f}(x)\ne f(x)]=O(m_i^{-c'\sqrt{m_i}})$ for all $x$. So, $f\in \ell$ with probability $1-O(2^{-c'\sqrt{m_i}\log_2(m_i)/2})$, in which case every element of $\ell$ will disagree with $f$ on at most $2^{m_i-c'\sqrt{m_i}\log_2(m_i)/2}$ elements. By lemma \ref{listDecodeLem}, the probability that this happens and there is an element of $\ell$ that agrees with $\tilde{f}$ on at least as many elements as $f$ does is $1-O(2^{-\sqrt[3]{n}})$. This algorithm recovers thus $f$ with probability $1-O(2^{-c'\sqrt{m}\log_2(m)/2})$, as desired.
\end{proof}

Given any sequences $m_i$ and $r_i$ such that $|r_i-m_i/2|=O(\sqrt{m_i})$ and the rate of $RM(m_i,r_i)$ is bounded below the capacity of the channel, there exists $c>0$ such that the rate of $RM(m_i+c\sqrt{m_i},r_i)$ is also bounded below the capacity of the channel, so the fact that we cannot recover a codeword drawn from $RM(m_i+c\sqrt{m_i},r_i)$ implies that $P_e(m_i,r_i,\epsilon)> 1/2-2^{-c\sqrt{m_i}/9}$ for all but finitely many $i$. Furthermore, $P_e(m,r,\epsilon)$ is nondecreasing in $r$, so that means that for any  sequences $m_i$ and $r_i$ such that the rate of $RM(m_i,r_i)$ is bounded below the capacity of the channel $P_e(m_i,r_i,\epsilon)= 1/2-2^{-\Omega(\sqrt{m_i})}$.

Also, note that while this is a bound on the accuracy with which we can determine the value of $f(x)$ from the values of $\tilde{f}$ other than $\tilde{f}(x)$, we can convert any algorithm for finding $f(x)$ given all values of $\tilde{f}$ into one that ignores $\tilde{f}(x)$ simply by setting the value of $\tilde{f}(x)$ in its input randomly. So, for any algorithm $L^\star$ that attempts to compute $f(0^m)$ given $\tilde{f}$, if the rate of the RM code is bounded above the capacity of the channel, we have 
\[\mathbb{P}[L^\star(\tilde{f})=f(0^m)]\le 1-\epsilon+O(2^{-c\sqrt{m}})-\Omega(\mathbb{P}[L^\star(\tilde{f})\ne \tilde{f}(0^m)])\]
for some constant $c>0$. In other words, given the entirety of $\tilde{f}$ we cannot determine the value of $f(x)$ with accuracy significantly greater than $1-\epsilon$, and the only way to attain an accuracy that high is to essentially always guess that $f(x)$ is $\tilde{f}(x)$. Thus there is no algorithm for recovering $f(x)$ that is significantly better than taking the value of $\tilde{f}(x)$.

\section{Generalization to BMS channels}\label{bms}
Now, consider a case where instead of flipping every bit with some fixed probability $\epsilon$, for each bit $X$ our channel generates $\epsilon\in [0,1/2]$ according to some probability distribution, selects $Z\in\mathbb{F}_2$ which is $1$ with probability $\epsilon$ and $0$ otherwise, and then outputs $(\epsilon_u,X_u+Z_u)$. We claim that our result can be generalized to a version saying that for any such channel, if $\limsup_{i\to\infty} {m_i \choose \le r_i}2^{-m_i}$ is strictly less than the capacity of the channel we can recover codewords in $RM(m,r)$ after putting them through the channel with probability $1-o(1)$. 

In order to show this, we would first observe that for any positive integer $k$, if we round $\epsilon$ to the nearest multiple of $1/k$ after receiving it from the channel that gives us a new channel of equal or lower capacity. However, in the limit as $k\to\infty$, the capacity of this channel goes to the capacity of the original channel, so for sufficiently large $k$ this channel still has capacity greater than $\limsup_{i\to\infty} {m_i \choose \le r_i}2^{-m_i}$. So, we can choose such a $k$, round all $\epsilon$ output by the channel to the nearest multiple of $1/k$, and then replace these epsilons with the average values of epsilon that round to those multiples of $1/k$. Thus, we can assume that there are only finitely many possible values of $\epsilon$.

Call the BMS channel $\mathcal{P}$, and define $P_e(m,r,\mathcal{P})$ as the appropriate analog of $P_e(m,r,\epsilon)$, i.e., the probability that we can determine the value of $f(0^m)$ correctly given the values of $\epsilon$ and $\tilde{f}(x)$ for all $x\ne 0^m$. 
The base case holds for BMS channels, and 
the core idea is the same as before, setting increasingly tight bounds on the bitwise error rate that we can attain by finding a set of subspaces $W_i$ that intersect at $V$, running recovery algorithms separately on each of them, and concluding that the target bit is the value that the majority of them decodings. However, now whether the algorithm will recover the target bit correctly using a subspace depends on the values of $\epsilon$ and $Z$ on that subspace rather than just depending on the values of $Z$. This complicates the argument using Fourier transforms to bound the probability that the restriction of the error vector to $V$ is bad. 

To make that more precise, let $k$ be the number of possible ordered pairs $(\epsilon_x,Z_x)$ for a bit $f(x)$. In order to define an appropriate basis for the Fourier transform, we need $k$ orthonormal functions from the space of ordered pairs $(\epsilon_x,Z_x)$ to $\mathbb{R}$, $\mathcal{X}^{(0)}$,...,$\mathcal{X}^{(k-1)}$. Also, we want $\mathcal{X}^{(0)}$ to be the constant function $1$ so that we can have terms that are independent of most of the inputs. That allows us to define $\mathcal{X}_{(S_1,...,S_{k-1})}$ for all disjoint $S_1,...,S_{k-1}\subseteq\{0,1\}^m$ so that for each $(\epsilon,Z)$, 
\[\mathcal{X}_{(S_1,...,S_{k-1})}(\epsilon,Z)=\prod_{i=1}^{k-1}\prod_{x\in S_i}\mathcal{X}^{(i)}(\epsilon_x,Z_x)\]

The orthonormality properties of the $\mathcal{X}^{(i)}$ imply that these are also orthonormal, and there are $k^m$ of them, so they form an orthonormal basis for the functions that take an input of the form $(\epsilon,Z)$ and return a real number. In this case, $Q_{m,m}$ is still symmetric under linear transformations, so if $\pi$ is such a transformation then $\langle Q_{m,m}, \mathcal{X}_{(S_1,...,S_{k-1})}\rangle =\langle Q_{m,m}, \mathcal{X}_{(\pi(S_1),...,\pi(S_{k-1}))}\rangle$ for all $S_1,...,S_{k-1}$. The fact that the $\mathcal{X}^{(i)}$ are orthogonal implies that for all $i\ne 0$, 
\begin{align}
    &E_{(\epsilon,Z)}[\mathcal{X}^{(i)}(\epsilon,Z)]\\
    &=E_{(\epsilon,Z)}[\langle \mathcal{X}^{(0)}(\epsilon,Z), \mathcal{X}^{(i)}(\epsilon,Z)\rangle]\\
    &=\langle\mathcal{X}^{(0)}, \mathcal{X}^{(i)}\rangle\\&=0
\end{align}
So, it is still the case that
\[Q_{\underline{m},m}=\sum_{S_1,..,S_{k-1}\subseteq\mathbb{F}_2^{\underline{m}}:S_i\cap S_j=\emptyset} \langle Q_{m,m}, \mathcal{X}_{(e_m(S_1),...,e_m(S_{k-1}))}\rangle\mathcal{X}_{(S_1,...,S_{k-1})}\]
because the contributions to $Q_{\underline{m},m}$ of all other terms cancel themselves out. Similarly, it is still the case that for a given $S_1,...,S_{k-1}\subseteq \mathbb{F}_2^m$ and a random linear transformation $\pi$, the probability that $\pi(S_i)$ is in $\mathbb{F}_2^{\underline{m}}$ for all $i$ is at most $2^{-(m-\underline{m})\dim(\cup_i S_i)}$. In a slight modification of the previous notation let $\overline{(S_1,...,S_{k-1})}$ be the set of all tuples of the form $(\pi(S_1),...,\pi(S_{k-1}))$ for a linear transformation $\pi$ and let $\mathbb{S}$ be a maximal list of $(k-1)$-tuples of subsets of $\mathbb{F}_2^{m}$ that are not equivalent under linear transformations. We still set $\mathcal{X}_{\overline{(S_1,...,S_{k-1})}}=\sum_{(S'_1,...,S'_{k-1})\in \overline{(S_1,...,S_{k-1})}} \mathcal{X}_{(S'_1,...,S'_{k-1})}$. Our proof that  $\lim_{i\to\infty} P_e(m_i,r_i,\mathcal{P})=O(2^{-\sqrt[3]{m_i}})$ is essentially the same as before. However, the details of the calculations underlying the lemma bounding the expectations of the fourth powers of $\mathcal{X}_{\overline{(S_1,...,S_{k-1})}}$ need to be updated as follows.

\begin{lemma}
Let $m$ be a positive integer and $S_1,...,S_{k-1}$ be disjoint subsets of $\mathbb{F}_2^m$. Also, let a BMS channel, $\mathcal{P}$, with finitely many possible values of $\epsilon$ and appropriate functions $\mathcal{X}^{(i)}$ be selected and $B=\max_{i,\epsilon,z}|\mathcal{X}^{(i)}(\epsilon,z)|$. Then $\langle\mathcal{X}^4_{\overline{(S_1,...,S_{k-1})}}\rangle\le 2^{2dm+8d^2}B^{2^{d+2}}$, where $d=\dim(\cup_i S_i)$
\end{lemma}

\begin{proof}
First, observe that
\[\langle\mathcal{X}^4_{\overline{S}}\rangle=\sum_{(S^{(1)}_1,...,S^{(1)}_{k-1}),...,(S^{(4)}_1,...,S^{(4)}_{k-1})\in\overline{S}} \langle\mathcal{X}_{(S^{(1)}_1,...,S^{(1)}_{k-1})}\cdot\mathcal{X}_{(S^{(2)}_1,...,S^{(2)}_{k-1})}\cdot\mathcal{X}_{(S^{(3)}_1,...,S^{(3)}_{k-1})}\cdot\mathcal{X}_{(S^{(4)}_1,...,S^{(4)}_{k-1})}\rangle\]
If there is any $x$ that is in exactly one of the $S^{(i)}_j$s then $\langle\mathcal{X}_{(S^{(1)}_1,...,S^{(1)}_{k-1})}\cdot\mathcal{X}_{(S^{(2)}_1,...,S^{(2)}_{k-1})}\cdot\mathcal{X}_{(S^{(3)}_1,...,S^{(3)}_{k-1})}\cdot\mathcal{X}_{(S^{(4)}_1,...,S^{(4)}_{k-1})}\rangle=0$ because the contributions of choices of $z$ and $\epsilon$ that take different values on $x$ cancel each other out. If there is no such $x$, then $\langle\mathcal{X}_{(S^{(1)}_1,...,S^{(1)}_{k-1})}\cdot\mathcal{X}_{(S^{(2)}_1,...,S^{(2)}_{k-1})}\cdot\mathcal{X}_{(S^{(3)}_1,...,S^{(3)}_{k-1})}\cdot\mathcal{X}_{(S^{(4)}_1,...,S^{(4)}_{k-1})}\rangle\le B^{4\sum_i |S_i|}\le B^{2^{d+2}}$. Also, $\cup_{i,j} S^{(i)}_j$ must contain $\dim(\cup_{i,j} S^{(i)}_j)$ points that are linearly independent. If there is no $x$ that is in exactly one of these sets, then each of these points must be in at least two of the $S^{(i)}_j$ with different values of $i$, which implies that $d\ge \dim(\cup_{i,j} S^{(i)}_j)/2$. There are at most $2^{2dm}$ $2d$-dimensional subspaces of $\mathbb{F}_2^m$, and there are at most $2^{2d^2}$ linear transformations of $(S_1,...,S_{k-1})$ contained in any such subspace. Hence $\langle\mathcal{X}^4_{\overline{S}}\rangle\le 2^{2dm+8d^2}B^{2^{d+2}}$, as desired.
\end{proof}

The combination of this change and the fact that the obvious bound on the number of equivalence classes of tuples $(S_1,...,S_{k-1})$ with $\dim(\cup_i S_i)=d$ is $k^{2^d}$ instead of $2^{2^d}$ changes the bound we get for $\mathbb{P}[P_e(\underline{m},m,r,\mathcal{P}|z')\ge 1/3]$ from $18P^{5/4}_{e}(m,r,\mathcal{P}) +9P_e(m,r,\mathcal{P})\left(\frac{4\log(64/\epsilon(1-\epsilon))}{\log(1/P_e(m,r,\mathcal{P}))}\right)^{m-\underline{m}}$ to $18P^{5/4}_{e}(m,r,\mathcal{P}) +9P_e(m,r,\mathcal{P})\left(\frac{c}{\log(1/P_e(m,r,\mathcal{P}))}\right)^{m-\underline{m}}$ for some other constant $c$ which depends on the parameters. However, this is never looser than the bound we would get for sufficiently unfavorable $\epsilon$, so the argument using that bound to prove that if $\limsup_{i\to\infty} {m_i \choose \le r_i}2^{-m_i}$ is strictly less than the capacity of the channel then there exists $c'>0$ such that $P_e(m_i,r_i,\mathcal{P})=O(m_i^{-c'\sqrt{m_i}})$ proceeds essentially unchanged. 

Finally, we need to address the part of the argument where we go from being able to recover each bit with accuracy $1-O(m_i^{-c'\sqrt{m_i}})$ to being able to completely recover the codeword with a high probability of success. This part can be left almost exactly unchanged. One may think that one should take the values of $\epsilon$ into account when deciding which element of $\ell$ to return, but ignoring them essentially just reduces the channel to a binary symmetric channel with noise $\mathbb{E}[\epsilon]$. The part of the argument showing that there is unlikely to be any codeword within $2^{m-c\sqrt{m}\log_2(m)/2}$ of the true codeword that agrees with $\tilde{f}$ on more elements than the true codeword does works for any $\epsilon\in (0,1/2)$, so this is good enough to establish that the algorithm will return the true codeword with probability $1-2^{-\Omega(\sqrt{m}\log(m))}$. The part about tightening the bounds on error probabilities further just uses our abilities to recover the codeword with probability $1-2^{-\Omega(\sqrt{m}\log(m))}$ and to go from a mostly accurate decoding to a completely accurate decoding with probability $1-O(2^{-\sqrt[3]{n}})$ as black boxes, so it is unchanged. Combining all of this would yield a theorem as follows:

 \begin{theorem}
Let $\mathcal{P}$ be a BMS channel and $m_1,m_2,...$ and $r_1,r_2,...$ be sequences of positive integers such that $\lim_{i\to\infty} m_i=\infty$ and $\limsup_{i\to\infty} {m_i \choose \le r_i}2^{-m_i}$ is strictly less than the capacity of $\mathcal{P}$. There exist $c,c',c''>0$ and an algorithm $RM\_reconstruction\_algorithm^\star$ such that in the limit as $i\to\infty$, if $f$ is randomly drawn from $RM(m_i,r_i)$ and $\tilde{f}=\mathcal{P}(f)$ then $RM\_reconstruction\_algorithm^\star(m_i,r_i,\mathcal{P},\tilde{f},c,c')$ returns $f$ with probability $1-O(2^{-2^{c''\sqrt{m}}})$.
\end{theorem}

The converse is essentially the same boosting argument used as a proof by contradiction ruling out the possibility of determining the value of $f(x)$ with nontrivial accuracy from the noisy versions of the other bits when the code's rate is greater than the channel capacity. The converse uses a variant of lemma \ref{errorBound1} in order to prove that we could start our boosting with $P_e(m,r,\epsilon)=1/2-2^{-c\sqrt{m}}$ instead of needing $P_e(m,r,\epsilon)=1/2-\Omega(1)$; however, both versions of the lemma just use lemma \ref{boostLemma1} as their source of information on the probability that the restriction of $Z$ to $V$ is bad. So, adapting the converse result that if the code's rate is bounded above the channel's capacity then $P_e(m_i,r_i,\epsilon)\ge 1/2-2^{-c\sqrt{m_i}}$ for some constant $c$ to a general BMS channel does not require any changes that were not already necessary to adapt the original argument to a general BMS channel. 

A little more precisely, to adapt the converse section to general BMS channels, we can modify the argument that we can assume that there are only finitely many possible values of $\epsilon$ by arguing that we can replace the actual BMS channel with one where each possible value of $\epsilon$ is replaced by $\lfloor k \epsilon\rfloor/k$ and argue that since the original channel is equivalent to the new channel with more noise added we must be able to recover the codewords at least as accurately on the new channel. We would need to replace all instances of $P_e(m,r,\epsilon)$ with $P_e(m,r,\mathcal{P})$, replace any use of lemmas in previous sections with use of appropriately generalized versions, replace the $4\log(64/\epsilon(1-\epsilon))$ term in the proof of lemma \ref{BMSBoostLem} with what results from the generalized version of lemma \ref{boostStepLem}, and replace the $e^{-108\ln(64/\epsilon(1-\epsilon))\sqrt[4]{m_i}}$ term with the expression necessary to ensure that the fraction that contained the $4\log(64/\epsilon(1-\epsilon))$ terms still evaluates to less than $\sqrt[4]{m_i}/27$. The observation that we can convert any algorithm for recovering $f(x)$ given the entirety of $\tilde{f}$ to one that ignores the value of $\tilde{f}(x)$ by setting the value of $\tilde{f}(x)$ in its input randomly generalizes to an observation that we can convert any algorithm for recovering $f(x)$ given the entirety of $\tilde{f}$ and $\epsilon$ to one that ignores the value of $\tilde{f}(x)$ by setting the value of $\tilde{f}(x)$ in its input randomly and the value of $\epsilon_x$ in its input arbitrarily. That implies that for any algorithm $L^\star$ that attempts to compute $f(0^m)$ given $\tilde{f}$ and $\epsilon$ and any possible value of $\epsilon_x$, $\epsilon'$, 
\[\mathbb{P}[L^\star(\tilde{f},\epsilon)=f(0^m)|\epsilon_{0^m}=\epsilon']\le 1-\epsilon'+O(2^{-c\sqrt{m}})-\Omega(\mathbb{P}[L^\star(\tilde{f})\ne \tilde{f}(0^m)|\epsilon_{0^m}=\epsilon'])\]
for some constant $c>0$. In other words, given the entirety of $\tilde{f}$ and $\epsilon$ we cannot determine the value of $f(x)$ with accuracy significantly greater than $1-\epsilon_x$, and the only way to attain an accuracy that high is to guess most of the time that $f(x)$ is $\tilde{f}(x)$ unless $\epsilon_x=1/2$.

\newcommand{\etalchar}[1]{$^{#1}$}


\end{document}